\documentclass[11pt]{article}
\usepackage[T1]{fontenc}
\usepackage[latin1]{inputenc}
\usepackage{amsmath,amsthm,amssymb}
\usepackage{gastex}
\usepackage{hyperref}
\usepackage{enumitem}
\usepackage{xspace}
\usepackage{setspace}
\usepackage{srcltx}
\usepackage{stmaryrd}

\def\titlerunning#1{\gdef\titrun{#1}}
\makeatletter
\def\author#1{\gdef\autrun{\def\and{\unskip, }#1}\gdef\@author{#1}}
\def\address#1{{\def\and{\\\hspace*{18pt}}\renewcommand{\thefootnote}{}%
\footnote {#1}}%
\markboth{\autrun}{\titrun}}
\makeatother
\def\email#1{e-mail: #1}
\def\subjclass#1{{\renewcommand{\thefootnote}{}%
\footnote{\emph{Mathematics Subject Classification (2010):} #1}}}
\def\keywords#1{\par\medskip
\noindent\textbf{Keywords.} #1}

\setenumerate[1]{label=$(\alph*)$,leftmargin=*}
\setenumerate[2]{label=$(\roman*)$,leftmargin=*}
\setitemize[1]{label=$\bullet$,leftmargin=*}
\setitemize[2]{label=$\circ$,leftmargin=*}
\setitemize[3]{label=$--$,leftmargin=*}

\newtheorem{Thm}{Theorem}[section]
\newtheorem{Prop}[Thm]{Proposition}
\newtheorem{Lemma}[Thm]{Lemma}
\newtheorem{Cor}[Thm]{Corollary}
\newtheorem{Remark}[Thm]{Remark}

\theoremstyle{definition}
\newtheorem{Def}[Thm]{Definition}

\numberwithin{equation}{section}

\def\freeze#1{\text{\underbar{\ensuremath{#1}}}}

\let\leq\leqslant
\let\geq\geqslant
\let\preceq\preccurlyeq
\let\succeq\succcurlyeq

\newcount\foo
\def\straight(#1,#2,#3,#4,#5){
  \node[Nw=.5,Nh=4,Nfill=y,Nmr=0,Nframe=n](node1)(#1,0){}
  \node[Nw=.5,Nh=4,Nfill=y,Nmr=0,Nframe=n](node2)(#2,0){}
  \drawedge[AHnb=0,linewidth=#5,ELside=#3,ELdist=1](node1,node2){#4}
}

\def\curvel(#1,#2,#3,#4,#5,#6,#7){
  \foo=#1\relax
  \advance\foo#4\relax
  \divide\foo2\relax
  \node[Nframe=n,Nw=0,Nh=0](A)(#1,#2){}
  \node[Nframe=n,Nw=0,Nh=0](B)(#4,#2){}
  \drawqbedge[AHnb=0,ELside=#7](A,\the\foo,#3,B){#6}
  \ifnum#2=#5\else
  \node[Nframe=n,Nw=0,Nh=0](A0)(#1,#5){}
  \node[Nframe=n,Nw=0,Nh=0](B0)(#4,#5){}
  \drawedge[AHnb=0,dash=1 0](A,A0){}
  \drawedge[AHnb=0,dash=1 0](B,B0){}
  \fi
}

\let\op=\llbracket
\let\cl=\rrbracket
\def\terms#1{\ensuremath{\mathcal{T}_{#1}}}
\def\pv#1{\ensuremath{{\sf#1}}}
\def\Om#1#2{\ensuremath{\overline\Omega_{#1}{\sf#2}}}

\def\omup#1#2#3{\ensuremath{\Omega^{#1}_{#2}{\sf#3}}}
\def\omc#1#2{\omup{\omega}{#1}{#2}}
\def\rk#1{\ensuremath{\mathop{\text{rank}}[{#1}]}}
\def\pj#1{\ensuremath{p_{\sf#1}}}

\def\clos#1{\ensuremath{\mathrm{cl}(#1)}}
\def\closv#1#2{\ensuremath{\mathrm{cl}_\pv{#1}(#2)}}
\def\ind#1{\ensuremath{\mathrm{ind} (#1)}}
\def\subrel#1#2{\mathrel{\mathop{#2}\limits_{#1}}}

\begin{document}


\baselineskip=17pt


\titlerunning{Normal forms for free aperiodic semigroups}

\title{McCammond's normal forms for free aperiodic semigroups
  revisited}

\author{J. Almeida %
  \and %
  J. C. Costa %
  \and %
  M. Zeitoun %
}

\date{}

\maketitle

\address{ %
  J. Almeida: %
  CMUP, Dep.\ Matemática, Faculdade de Ciências, Universidade do
  Porto, Rua do Campo Alegre 687, 4169-007 Porto, Portugal; %
  \email{jalmeida[AT]fc.up.pt} %
  \and %
  J. C. Costa: %
  CMAT, Dep.\ Matemática e Aplicações, Universidade do Minho, Campus
  de Gualtar, 4700-320 Braga, Portugal; %
  \email{jcosta[AT]math.uminho.pt} %
  \and %
  M. Zeitoun: %
  LaBRI, Université Bordeaux \& CNRS UMR 5800, 351 cours de la
  Libération, 33405 Talence Cedex, France; %
  \email{mz[AT]labri.fr} %
}


\begin{abstract}
  This paper revisits the solution of the word problem for
  $\omega$-terms interpreted over finite aperiodic semigroups,
  obtained by J.~McCammond. The original proof of correctness of
  McCammond's algorithm, based on normal forms for such terms, uses
  McCammond's solution of the word problem for certain Burnside
  semigroups. In this paper, we establish a new, simpler, correctness
  proof of McCammond's algorithm, based on properties of certain
  regular languages associated with the normal forms. This method
  leads to new applications.

\end{abstract}

\section{Introduction}
\label{sec:introd}

An $\omega$-term is a formal expression obtained from letters of an
alphabet $X$ using two operations: the binary, associative,
concatenation and the unary $\omega$-power. Any $\omega$-term $\alpha$
can be given a natural interpretation on a finite semigroup $S$ as a
mapping $\alpha_S:S^X\to S$, as follows: each letter $x$ of $X$ is
interpreted as the mapping sending each element of $S^X$ to its image
on~$x$, the concatenation is viewed as the semigroup multiplication,
while the $\omega$-power is interpreted as the unary operation which
sends each element of $S$ to its unique idempotent power. The
$\omega$-word problem for a class $\mathcal{C}$ of finite semigroups
consists in deciding whether two $\omega$-terms have the same
interpretation over every semigroup of~$\mathcal{C}$.

One motivation for considering the $\omega$-word problem is that its
decidability is one of the requirements of a property of
pseudovarieties (classes of finite semigroups closed under taking
subsemigroups, homomorphic images, and finite direct products) called
\emph{tameness}, introduced by the first author and
Steinberg~\cite{Almeida:1999b,Almeida&Steinberg:2000b,Almeida&Steinberg:2000a}
to solve the decidability problem for iterated semidirect products of
pseudovarieties. In spite of its limitations for that purpose under
current knowledge, tameness remains a property of interest which has
also been used to solve membership problems involving other types of
operators~\cite{Almeida&Costa&Zeitoun:2004}
(see~\cite{Almeida&Costa&Zeitoun:2005b}
and~\cite[Section~3.7.3]{Rhodes&Steinberg:qt} for a discussion). A
difficult problem occurring in computer science is related to a weak
form of tameness~\cite{Almeida99:SomeAlgorProbl}. It asks if it is
possible to separate two given regular languages by a language
recognized by a semigroup of a given pseudovariety. For the
pseudovariety \pv A of all aperiodic (or group-free) semigroups,
recognizing exactly first-order definable
languages~\cite{Schutzenberger:1965,McNaughton&Papert:71}, it amounts
to finding a first-order formula holding on (all words of) one
language, and whose negation holds on the other one. It was solved
algebraically by Henckell~\cite{Henckell:1988}, and by simple combinatorial methods
by Place and the third author~\cite{PZ:lics14}.

The $\omega$-word problem has been solved for some pseudovarieties. The case
of the pseudovariety of all $\mathcal{J}$-trivial semigroups, solved by the
first author in~\cite{almeida:1990}, constitutes a classical example. Another
remarkable example, achieved by McCammond~\cite{McCammond:1999a}, is given by
the pseudovariety \pv A.  Recently, an alternative algorithm for deciding the
$\omega$-word problem for \pv A has been proposed
in~\cite{HuschenbettKufleitner:LIPIcs:2014:4472}. It is based on
Ehrenfeucht-Fra\"\i{}ss\'e games played on representations of $\omega$-words,
and this approach makes it possible to obtain an \textsc{Exptime} upper bound
for this decision problem. One should note however that the correctness proof
of this new algorithm itself relies on McCammond's algorithm.

The $\omega$-word problem has been solved for other pseudovarieties. It has
been obtained by the second author~\cite{Costa:2001} for the pseudovariety of
local semilattices. The first and third
authors~\cite{Almeida&Zeitoun:AutomTheorApproac:2007} solved the $\omega$-word
problem for the pseudovariety of $\mathcal{R}$-trivial~semigroups, and their
techniques have been adapted for the pseudovariety \pv{DA}, which consists of all finite
semigroups whose regular $\mathcal{J}$-classes are aperiodic
semigroups~\cite{Moura:2011:a}. Recently, the second
author~\cite{Costa:Canonical-forms-free-k-semigroups:2014:a} has applied
techniques similar to the ones of this paper to show decidability of the
$\omega$-word problem for the pseudovariety of all finite semigroups.

Unlike the cases of local
semilattices~\cite{Costa&Teixeira:2005,Costa&Nogueira:2009} and
$\mathcal{R}$-trivial
semigroups~\cite{Almeida&Costa&Zeitoun:2004,Almeida&Costa&Zeitoun:2005b},
there is of yet no published proof of tameness of \pv A, but the above
mentioned solution of the $\omega$-word problem for \pv A is a step forward in
that direction. McCammond's solution~\cite{McCammond:1999a} consists in the
reduction of arbitrary $\omega$-terms to a certain normal form. McCammond then
goes on to show that different $\omega$-terms in normal form cannot have the
same interpretation over~\pv A, which he does by invoking his results on free
Burnside semigroups~\cite{McCammond:1991}.

\subsection*{Contributions} We give an alternative proof of
McCammond's normal form theorem for $\omega$-terms over~\pv
A, which is independent of the theory of free Burnside
semigroups. Our approach consists in associating to each $\omega$-term $\alpha$
a decreasing sequence of regular languages $(L_n[\alpha])_n$, whose key
property is that, if $\alpha$ is in McCammond's normal form, then $L_n[\alpha]$
is ultimately star-free. Another crucial element in the proof is the
fact that if $\alpha$ and $\beta$ are $\omega$-terms in normal form and
$L_n[\alpha]\cap L_n[\beta]\neq\emptyset$ for all $n$, then $\alpha=\beta$.

This new approach, and particularly the fact that the languages
$L_n[\alpha]$ are star-free, also yields new applications on the
structure of the free pro-\pv A semigroup. Some elements of this
semigroup, called \emph{$\omega$-words}, have a nice form: they can
actually be represented by an $\omega$-term. We show that in the free
pro-\pv A semigroup, every factor of an $\omega$-word is also an
$\omega$-word. In turn this result is a central piece
in~\cite{Almeida&Costa&Zeitoun:2009}, whose main result provides a
characterization of $\omega$-words in the free pro-\pv A semigroup.

The paper is organized as follows. In Section~\ref{sec:preliminaries},
we review background material, including the description of
McCammond's normal form. We introduce term expansions and the
languages $L_n[\alpha]$ in Section~\ref{sec:term-expansions}, and we prove
some of their basic properties. Section~\ref{sec:combinatorial-lemmas}
is mainly devoted to the proof of a combinatorial and central lemma,
about $\omega$-terms whose $\omega$-powers are not nested. In
Section~\ref{sec:omega-word-problem-over-A}, we present the main
properties of the languages $L_n[\alpha]$ and the alternative proof of
uniqueness of McCammond's normal forms for $\omega$-terms over~\pv A.  In
Section~\ref{sec:star-freeness-Ln} we establish the star-freeness of
$L_n[\alpha]$ for $\alpha$ in normal form and $n$ large enough. Finally, we
investigate in Section~\ref{sec:factors-of-omega-words} other
properties of the languages $L_n[\alpha]$, and derive some~applications.

\section{Preliminaries}
\label{sec:preliminaries}

In this section we briefly recall the basic definitions and results
that will be used throughout the paper. The reader is referred
to~\cite{Almeida:1994a,Rhodes&Steinberg:qt} for general background,
and to~\cite{Almeida:2003c} for a quick introduction to the
classical theories of pseudovarieties, regular languages and profinite
semigroups. For further details about combinatorics on words,
see~\cite{Lothaire:1983,Lothaire:2001}.

\subsection{Words}
\label{sec:words}

In the following, $X$ denotes a finite nonempty alphabet. The
free semigroup (resp.\ the free monoid) generated by $X$
is denoted by $X^+$ (resp.\ by $X^*$). The length of a word $u \in X^*$
is denoted by $|u|$. Given words $u$ and $v$, we write $u\preceq v$ if $u$
is a prefix of~$v$ and $u\prec v$ if $u\preceq v$ and $u\neq v$.
If $v=uw$, we denote by $u^{-1}v$ the suffix $w$ of $v$.  When
$w=xyz=x'y'z'$, we say that the factors $y$ and $y'$ of $w$ are
\emph{synchronized in $w$} if $x=x'$ and $z=z'$ (whence $y=y'$).  They
\emph{overlap} if $x\preceq x'\prec xy$ or~$x'\preceq x\prec x'y'$.
They \emph{overlap on (at least) $k>0$ positions} if in addition $y=u_1vu_2$
and $y'=u'_1v'u'_2$ where $|v|=|v'|=k$ and $v,v'$ are synchronized in~$w$.

The following result is known as Fine and Wilf's Theorem
(see~\cite{Lothaire:1983,Lothaire:2001}).

\begin{Prop}[Fine and Wilf's Theorem]
  \label{prop:fine:wilf}
  Let $u,v\in X^+$. If two powers $u^k$ and $v^\ell$ of $u$ and $v$
  have a common prefix of length at least $|u|+|v|-\gcd(|u|,|v|)$,
  then $u$ and $v$ are powers of the same~word.\qed
\end{Prop}

A \emph{primitive word} is a word that cannot be written in the form
$u^n$ with $n>1$. Two words $w$ and $z$ are \emph{conjugate} if one
can write $w=uv$ and $z=vu$, where $u,v\in X^*$. All conjugates of a
primitive word are also primitive. Let an order be fixed for the
letters of the alphabet~$X$. A \emph{Lyndon word} is a primitive word
that is minimal, with respect to the lexicographic ordering, in its
conjugacy class. We recall a property following from
\cite[Prop.~5.1.2]{Lothaire:1983}.

\begin{Lemma}
  \label{l:Lyndon-less-than-suffices}
  If $t\in X^*$ is both a prefix and a suffix of a Lyndon word $w$,
  then either $t$ is the empty word, or $t$ is the word $w$ itself.
\end{Lemma}

\subsection{\texorpdfstring{Pseudowords and
    $\omega$-words}{Pseudowords and omega-words}}
\label{sec:pseudowords--words}

In this paper, we deal with the pseudovariety~\pv A of all finite
aperiodic, or group-free, semigroups. These are the
finite semigroups $T$ for which there exists some integer $n>0$ such
that $s^n=s^{n+1}$ for every $s\in T$. We write \pv S for the class
of all finite~semigroups.

Given a pseudovariety \pv V, we denote by \Om XV the free pro-\pv V
semigroup over~$X$ (see~\cite{Almeida:2003c} for its construction and
main properties). We briefly recall here some of its properties needed
in the paper. First, \Om XV is a compact topological semigroup whose
elements are called \emph{pseudowords} over~\pv V. For $\pv V=\pv S$
or $\pv A$, the free semigroup $X^+$ embeds in $\Om XV$ and is dense
in \Om XV. For $L\subseteq X^+$, we denote by $\clos L$, resp.\
$\closv AL$ its closure in \Om XS, resp.\ in \Om XA. There is a unique
continuous homomorphism from $\Om XS$ to $\Om XA$ sending each $x\in
X$ to itself, and we denote it by \pj A. Note that $\pj A(\clos
L)=\closv AL$.

Given $z\in\Om XV$, the closed subsemigroup of \Om XV generated by $z$
contains a single idempotent denoted by $z^\omega$, which is the limit
of the sequence $z^{n!}$. Note that $zz^{\omega}=z^{\omega}z$. We set
$z^{\omega+1}=zz^{\omega}$. In \Om XA, we have
$z^{\omega+1}=z^\omega$. For $\alpha,\beta\in\Om XS$, we say that \pv
A satisfies $\alpha=\beta$ if $\pj A(\alpha)=\pj A(\beta)$. For
example, \pv A satisfies $z^{\omega+1}=z^\omega$ for all $z\in\Om XS$.

A \emph{unary semigroup} is an algebra $(S,\cdot,\tau)$, with $\cdot$
binary and associative and $\tau$ unary. A free pro-\pv V semigroup
has a natural structure of unary semigroup, where $\tau$ is
interpreted as the $\omega$-power. We denote by \omc XV the unary
subsemigroup of \Om XV generated by $X$, whose elements are called
\emph{$\omega$-words} over $\pv V$. Each $\omega$-word has a
representation by a formal term over $X$ in the signature
$\{\cdot,\omega\}$, called an \emph{$\omega$-term}. We do not
distinguish between $\omega$-terms that only differ in the order in
which multiplications are to be carried
out.
Finally, let $\terms{X}$ be the unary semigroup of
$\omega$-terms, which is freely generated by~$X$ as a unary semigroup.
Sometimes, it will be useful to consider also the empty
  $\omega$-term, which is identified with the empty word.

\subsection{\texorpdfstring{The $\omega$-word problem
    for \pv A}{The omega-word problem for A}}
\label{sec:mccamm-norm-form}

McCammond~\cite{McCammond:1999a} represents $\omega$-terms over~$X$ as
nonempty well-parenthesized words over the alphabet
$Y=X\uplus\{\mathord{(},\mathord{)}\}$, which do not have $()$ as a
factor. The $\omega$-term associated with such a word is obtained by
replacing each matching pair of parentheses $(*)$ by $(*)^\omega$. For
example, the parenthesized word $((a)b)$ represents the $\omega$-term $(a^\omega
b)^\omega$. Conversely, every $\omega$-term over~$X$ determines a unique
well-parenthesized word over $Y$. We identify $\terms{X}$ with the set
of these well-parenthesized words over~$Y$. From hereon, we will
usually refer to an $\omega$-term meaning its associated word over $Y$. In
particular, there is a natural homomorphism of unary semigroups
$\epsilon:\terms{X}\to\omc XA$ that fixes each $x\in X$ when we view $X$ as a
subset of $\terms{X}$ and \omc XA in the natural way. To avoid
ambiguities in the meaning of the parentheses, we write $\epsilon[w]$ for the
image of $w\in \terms{X}$ under~$\epsilon$.

The $\omega$-word problem for~\pv A (over $X$) consists in deciding
whether two given elements of $\terms{X}$ have the
same image under~$\epsilon$. This problem was solved by McCammond by
effectively transforming any $\omega$-term into a certain
\emph{normal form} with the same image under~$\epsilon$, and by
proving that two $\omega$-terms in normal form with the same image
under~$\epsilon$ are necessarily equal. In order to describe the
normal form, let us fix a total ordering on the alphabet $X$, and
extend it to $Y=X\cup\{\mathord{(},\mathord{)}\}$ by letting
${(}<x<{)}$ for all $x\in X$. The \emph{rank} of an $\omega$-term
$\alpha$ is the maximum number $\rk \alpha$ of nested parentheses in
it.

McCammond's normal form is defined recursively. \emph{Rank~0 normal
  forms} are the words from~$X^*$. Assuming that rank~$i$
normal forms have been defined, a \emph{rank~$i+1$ normal form
  ($\omega$-term)} is an $\omega$-term of the form
$$\alpha_0(\beta_1)\alpha_1(\beta_2)\cdots\alpha_{n-1}(\beta_n)\alpha_n,$$
where the $\alpha_j$ and $\beta_k$ are $\omega$-terms such that the
following conditions hold:
\begin{enumerate}
\item\label{item:rank-i-1} each $\beta_k$ is a Lyndon word of
  rank~$i$;
\item\label{item:rank-i-2} no intermediate $\alpha_j$ is a prefix of a
  power of $\beta_j$ or a suffix of a power of~$\beta_{j+1}$;
\item\label{item:rank-i-3} replacing each subterm $(\beta_k)$ by
  $\beta_k\beta_k$, we obtain a rank~$i$ normal form;
\item\label{item:rank-i-4} at least one of the properties
  \ref{item:rank-i-2} and \ref{item:rank-i-3} fails if we remove
  from~$\alpha_j$ a prefix $\beta_j$ (for $0<j$) or a suffix $\beta_{j+1}$
  (for~$j<n$).
\end{enumerate}

For instance, if the letters $a,b\in X$ are such that $a<b$, then the
terms $(a) ab(b)$, $b(ab)abaa(a)aaab(aab)$ and $((a) ab(b)ba)(a)
ab(b)$ are in normal form.

McCammond's procedure to transform an arbitrary $\omega$-term into one in
normal form, while retaining its value under $\epsilon$, consists in applying
elementary changes determined by the following rewriting rules:
\begin{xalignat*}{2}
&1.~((\alpha))\rightleftarrows(\alpha)&
&4R.~ (\alpha)\alpha\rightleftarrows(\alpha)\\[-1ex]
&2.~ (\alpha^k)\rightleftarrows(\alpha)&
&4L.~\, \alpha(\alpha)\rightleftarrows(\alpha)\\[-1ex]
&3.~ (\alpha)(\alpha)\rightleftarrows(\alpha)&
&5.~\ \;\,(\alpha\beta)\alpha\rightleftarrows\alpha(\beta\alpha)
\end{xalignat*}
We call the application of a rule of type 1--4 from left to right
(resp.\ from right to left) a \emph{contraction} (resp.\ an
\emph{expansion}) of that type.

Since all the rules are based on identities of unary semigroups that
are valid in~\pv A (in fact, all but those of type~4 are valid in~\pv
S), it follows that the elementary changes preserve the value of the
$\omega$-term under~$\epsilon$. Hence McCammond's algorithm does indeed transform
an arbitrary $\omega$-term into one in normal form with the same image
under~$\epsilon$. We don't describe here McCammond's procedure because
usually we will work with $\omega$-terms already in normal form. The reader
interested in the algorithm is referred to the original
paper~\cite{McCammond:1999a} or to~\cite{Almeida&Costa&Zeitoun:2009}
for a more condensed description.

\section{\texorpdfstring{Expansions of $\omega$-terms}{Expansions of omega-terms}}
\label{sec:term-expansions}

The main tools of this paper is to associate to any $\omega$-term
$\alpha$ a decreasing sequence $\bigl(L_n[\alpha]\bigr)_n$ of regular
languages. Informally, for $n>0$, the language $L_n[\alpha]$ is
obtained from $\alpha$ by replacing each $\omega$-power by a power of
exponent at least $n$. That is, $L_n[\alpha]$ is the language obtained
from $\alpha$ by replacing each ``$\omega$'' by ``$\geq n$'', where we
set $L^{\geq n}=L^*L^n$ for $L\subseteq X^{+}$.

Clearly, the sequence $\bigl(L_n[\alpha]\bigr)_{n}$ is decreasing, and
$\epsilon[\alpha]$ belongs to the topological closure, in \Om XA, of each $L_n[\alpha]$. The key
result (Theorem~\ref{t:star-free} below) is that $L_n[\alpha]$ is star-free
for $\alpha$ in normal form and $n$ large enough.

We now formally define $L_n[\alpha]$, first defining intermediate
expansions that only unfold the outermost $\omega$-powers enclosing
subterms of maximum  rank. The main differences between this definition
and McCammond's ``rank~$i$ expansions''
\cite[Definition~10.5]{McCammond:1999a} are that we require the
exponents to be beyond a fixed threshold and we do not require that
the $\omega$-terms be in normal~form.

\begin{Def}[Word expansions]
  Let $n$ be a positive integer. For a word $\alpha\in X^*$, we let
  $E_n[\alpha]=\{\alpha\}$. Let $i\geq0$. For an $\omega$-term
  \begin{equation}
    \label{eq:an-omega-word}
    \alpha=\gamma_0(\delta_1)\gamma_1\cdots(\delta_r)\gamma_r
    \ \parbox[t]{.6\textwidth}{where all $\delta_k$ are $\omega$-terms of
      rank~$i$ and all $\gamma_j$ are either empty, or $\omega$-terms of rank at most~$i$,}
  \end{equation}
  we let
  $$E_n[\alpha] =\{\gamma_0\delta_1^{n_1}\gamma_1\cdots\delta_r^{n_r}\gamma_r: \ n_1,\ldots,n_r\geq n\}.$$
  For a set $W$ of $\omega$-terms, we let $E_n[W]=\bigcup_{\alpha\in W}E_n[\alpha]$. We then let
  $$L_n[\alpha]=E_n^{\rk \alpha}[\alpha],$$
  where $E_n^k$ is the $k$-fold iteration of the operator $E_n$. For a set $W$ of
  $\omega$-terms, we let $L_n[W]=\bigcup_{\alpha\in W}L_n[\alpha]$.
\end{Def}
For example, let $\alpha=(a^\omega b)^\omega$ and $n=3$. We have $\rk
\alpha=2$, so $L_3[\alpha]=E_3^2[\alpha]$. Then,
$E_3[\alpha]=\{(a^\omega b)^p\mid p\geq 3\}$ and
$L_3[\alpha]=(a^*a^3b)^*(a^*a^3b)^3$.

\begin{Lemma}\label{l:Ln}
  The following formulas hold:
  \begin{enumerate}
  \item\label{item:Ln-1} for $\omega$-terms $\alpha$ and $\beta$,
    $$E_n[\alpha\beta]=
    \begin{cases}
      E_n[\alpha]\,E_n[\beta] & \mbox{if $\rk \alpha=\rk \beta$}\\
      \alpha\,E_n[\beta] & \mbox{if $\rk \alpha<\rk \beta$}\\
      E_n[\alpha]\,\beta & \mbox{if $\rk \alpha>\rk \beta$};
    \end{cases}
    $$
  \item\label{item:Ln-1b} for an $\omega$-term $\alpha$, $L_n[\alpha]=L_n[E_n[\alpha]]$;
  \item\label{item:Ln-2} for sets $U$ and $V$ of $\omega$-terms, we have $L_n[UV]=L_n[U]\,L_n[V]$;
  \item\label{item:Ln-3} for a factorization $\alpha=\gamma_0(\delta_1)\gamma_1\cdots(\delta_r)\gamma_r$ of an $\omega$-term as in~\eqref{eq:an-omega-word}:
    $$L_n[\alpha]=L_n[\gamma_0]\,L_n[(\delta_1)]\,L_n[\gamma_1]\cdots
    L_n[(\delta_r)]\, L_n[\gamma_r];$$
  \item\label{item:Ln-4} for an $\omega$-term $\alpha$,
    $L_n[(\alpha)]= L_n[\alpha]^*L_n[\alpha]^n$.
  \end{enumerate}
\end{Lemma}

\begin{proof}
  \ref{item:Ln-1} is immediate from the definition of the operator
  $E_n$. For~\ref{item:Ln-1b}, since $E_n[\alpha]$ is a set of
  $\omega$-terms whose rank is $\rk\alpha-1$, we
  have
  $$L_n[E_n[\alpha]]=E_n^{\rk \alpha-1}[E_n[\alpha]]=E_n^{\rk
    \alpha}[\alpha]=L_n[\alpha].$$

  We first establish \ref{item:Ln-2} when the rank of elements of
  $U\cup V$ is bounded by some $m>0$, proceeding by induction on~$m$.
  For $\omega$-terms $\alpha$ and $\beta$ of rank at most $m$, we have
  \begin{align*}
    L_n[\alpha\beta] \subrel{\hbox{\scriptsize\ref{item:Ln-1b}}}{=}L_n[E_n[\alpha\beta]]&=\left\{
      \begin{array}[c]{ll}
        L_n[E_n[\alpha ]E_n[\beta]]&\mbox{if $\rk\alpha =\rk\beta$}\\
        L_n[\alpha E_n[\beta]]&\mbox{if $\rk\alpha <\rk\beta$}\\
        L_n[E_n[\alpha]\beta]&\mbox{if $\rk\alpha >\rk\beta$}
      \end{array}\right.\text{ by~\ref{item:Ln-1}}\\
    &=L_n[\alpha ]\,L_n[\beta]\text{ by induction hypothesis and \ref{item:Ln-1b}.}
  \end{align*}
  To conclude the induction step, note that
  \begin{equation}
    \label{eq:LnUV}
    L_n[UV]
    =\bigcup_{\alpha\in U,\,\beta\in V}L_n[\alpha\beta]
    =\bigcup_{\alpha\in U,\,\beta\in V}L_n[\alpha]L_n[\beta]
    =L_n[U]L_n[V].
  \end{equation}
  This shows in particular that $L_n[\alpha\beta] = L_n[\alpha]L_n[\beta]$ for all
  $\omega$-terms $\alpha$ and $\beta$, so that \eqref{eq:LnUV} still holds for
  arbitrary sets $U$ and $V$, which establishes~\ref{item:Ln-2}.

  Property \ref{item:Ln-3} follows from \ref{item:Ln-2} by induction
  on the number of factors. For~\ref{item:Ln-4}, we have
  $$L_n[(\alpha)]
  \subrel{\hbox{\scriptsize\ref{item:Ln-1b}}}{=}
  L_n[E_n[(\alpha)]]
  =\bigcup_{m\geq n}L_n[\alpha^m]
  \subrel{\hbox{\scriptsize\ref{item:Ln-2}}}{=}
  \bigcup_{m\geq n}L_n[\alpha]^m
  =L_n[\alpha]^*L_n[\alpha]^n.\popQED
  $$
\end{proof}

In case $\alpha$ is a rank ($i+1$) $\omega$-term in normal form, the
elements of $E_1[\alpha]$ are precisely McCammond's ``rank~$i$
expansions of $\alpha$''. Since Lemma~10.7 of~\cite{McCammond:1999a}
states that every such rank~$i$ expansion of $\alpha$ remains in
normal form and since $E_1[\alpha]\supseteq E_2[\alpha]\supseteq
E_3[\alpha]\supseteq\cdots$, we obtain the following result.

\begin{Lemma}
  \label{l:McCammond-Lemma-10.7}
  If $\alpha$ is an $\omega$-term in normal form, then all $\omega$-terms of
  $E_n^k[\alpha]$ for $n,k\geq 1$ are also in normal form.\qed
\end{Lemma}

We now associate to each term $\alpha$ a parameter $\mu[\alpha]$ playing an
important role in this paper. First define the \emph{length} of an
$\omega$-term $\alpha$ as the length of the corresponding well-parenthesized
word over $Y$, and denote it $|\alpha|$. For an $\omega$-term $\alpha$ as
in~\eqref{eq:an-omega-word}, the factors of $\alpha$ of the form
$(\delta_j)\gamma_j(\delta_{j+1})$ are called \emph{crucial portions} of $\alpha$.

\begin{Def}
  Let $\alpha$ be an $\omega$-term. In case $\alpha\in X^+$, let
  $\mu[\alpha]=0$. Otherwise,~let
  $$\mu[\alpha]=2^{\rk \alpha}\max\{|\beta|:
  \beta\text{ is a crucial portion of $\alpha^2$}\}.$$
\end{Def}

It is important to point out the following simple observation.

\begin{Lemma}
  \label{l:mu-vs-expansions}
  If $\alpha$ is an $\omega$-term and $\bar \alpha\in E_n[\alpha]$, then $\mu[\bar \alpha]\leq\mu[\alpha]$.
\end{Lemma}

\begin{proof}
  The statement is clear if $\rk\alpha\leq1$. Otherwise, $\mu[\bar
  \alpha]=2^{\rk{\bar \alpha}}|\bar \beta|$ for some crucial portion
  $\bar \beta$ of $\bar \alpha^2$. Since
  $2^{\rk{\alpha}}=2\cdot2^{\rk{\bar \alpha}}$, it suffices to show
  that there exists a crucial portion $\beta$ of $\alpha^2$ such that
  $|\bar \beta|\leq2|\beta|$. Since $\bar \alpha^2\in E_n(\alpha^2)$
  by Lemma~\ref{l:Ln}, $\bar \beta$ is a factor of either some
  $\delta\gamma\delta'$ where $(\delta)\gamma(\delta')$ is a crucial
  portion of $\alpha^2$, or of some $\delta\delta$ where $(\delta)$ is
  a factor of $\alpha$ of maximum rank. In the first case, choose
  $\beta=(\delta)\gamma(\delta')$ so that $|\bar \beta|\leq|\beta|$.
  In the second one, take for $\beta$ any crucial portion of
  $\alpha^2$ involving $(\delta)$. Then $|\bar
  \beta|\leq2|\delta|\leq2|\beta|$, as required.
\end{proof}

For an $\omega$-term $\alpha$ of positive rank, we distinguish the innermost,
rank~1, parentheses as new letters $\op$ and $\cl$. We extend the
ordering over the enlarged alphabet $X\cup\{{\op},{\cl}\}$ by letting
${\op}<x<{\cl}$ ($x\in X$). Under this interpretation, we view $\alpha$ as an
$\omega$-term over $X\cup\{{\op},{\cl}\}$, denoted $\freeze \alpha$ and called the
\emph{freeze} of~$\alpha$.\label{def:freeze}
\begin{Remark}
  \label{remark:freezed-parentheses}
  The freeze~$\freeze \alpha$ of an $\omega$-term $\alpha$ satisfies
  the relations $\rk{\freeze \alpha}=\rk \alpha-1$, and $\mu[\freeze
  \alpha]\leq\mu[\alpha]/2$. Moreover, if $\alpha$ (resp.\ its crucial
  portions) is in normal form, then so is~$\freeze\alpha$ (resp.\ so
  are its crucial portions).
\end{Remark}

\section{A synchronization result}
\label{sec:combinatorial-lemmas}

We prove in this section a synchronization result for $\omega$-terms of rank~1.

\begin{Prop}
  \label{p:synchro}
  Let $\alpha=u_0(v_1) u_1\cdots (v_r) u_r$ and $\beta=z_0(t_1) z_1\cdots (t_s) z_s$ be
  two $\omega$-terms of rank~1 in normal form, and let
  $n\geq\max\{\mu[\alpha],\mu[\beta]\}$. Let
  $$w=u_0v_1^{n_1} u_1\cdots v_r^{n_r} u_r=z_0t_1^{m_1} z_1\cdots t_s^{m_s}
  z_s\in L_n[\alpha]\cap L_n[\beta].$$ Then $r=s$, and for all $i$, $u_i=z_i$,
  $n_i=m_i$ and $v_i=t_i$. In particular, $\alpha=\beta$.
\end{Prop}

The remainder of this section is devoted to the proof of
Proposition~\ref{p:synchro}. For a \emph{factorization}
$w=u_0^{}v_1^{n_1} u_1^{}\cdots v_r^{n_r}u_r$ (which will be clear
from the context), we denote by $w[i]$ the word $u_0^{}v_1^{n_1}
\cdots u_{i-1}^{}v_i^{n_i}$---empty for $i=0$, by convention.

We shall use the following synchronization property:
if two powers of Lyndon words have a large common factor, then the
Lyndon words are equal, and the common factor starts in the
\emph{same} position in both~of~them.

\begin{Lemma}\label{l:FW}
  Let $u$ and $v$ be Lyndon words, and let $w$ be a factor of both a
  power of~$u$ and a power of~$v$: $u^m=xwy$ and $v^n=zwt$. If
  $|w|\geq|u|+|v|$, then $u=v$, and there is a factorization $w=w_1w_2$
  such that $xw_1,zw_1\in u^*$.
\end{Lemma}

\begin{proof}
  The hypothesis implies that $w$ is a prefix of both a power of a
  conjugate $\tilde u=u_2u_1$ of $u=u_1u_2$ and of a power of a conjugate
  $\tilde v=v_2v_1$ of $v=v_1v_2$. By Fine and Wilf's Theorem
  (Proposition~\ref{prop:fine:wilf}) $\tilde u$ and $\tilde v$ are
  powers of the same word. Since $\tilde u$ and $\tilde v$ are
  primitive, they are equal, hence the Lyndon words in their class,
  $u$ and $v$ respectively, are also equal.

  By symmetry, one may assume that $u_1\neq1$ and $u_1\preceq v_1$. Since $u_2u_1$ and
  $v_2v_1$ are conjugates of the same primitive word $u_1u_2=v_1v_2$,
  they are of the form $rs$ and $sr$ with $r=u_1^{-1}v_1^{}$ and
  $s=v_2u_1$. Since they are equal, we obtain $r,s\in p^*$ for some word
  $p$ by~\cite[Prop.~1.3.2]{Lothaire:1983}, and since they are
  primitive, we get $r=1$ or $s=1$, whence $u_1=v_1$ and
  $u_2=v_2$. Moreover, $x=u^ku_1$ and $z=u^\ell v_1$. Therefore,
  $w_1=u_2=v_2$ meets the requirements of the~lemma.
\end{proof}

\begin{Remark}
  Let $\alpha=u_0(v_1) u_1\cdots (v_r) u_r$ be an $\omega$-term of
  rank 1. Let $z$ be a nonempty word, and let $m\geq\mu[\alpha]$.
  Then, for each $i\in\{1,\ldots,r\}$, we have
  \begin{equation}
    \label{eq:overlap}
    |z^m|\geq|u_{i-1}u_i|+|v_iz|.
  \end{equation}
  Indeed, one may assume by symmetry that $|u_i| \geq|u_{i-1}|$. Let 
  $\beta$ be $(v_i)u_i(v_{i+1})$ if $i<r$, or $(v_r)u_ru_0(v_1)$ if $i=r$.
  Since $\beta$ is a crucial portion of $\alpha^2$, we have $m\geq\mu[\alpha]\geq2|\beta|
  \geq2|u_i|+|v_i|+1\geq|u_{i-1}v_iu_i|+1$, so
  $|z^m|\geq|z^{|u_{i-1}v_iu_i|+1}|=(|u_{i-1}v_iu_i|+1)|z|\geq|u_{i-1}u_i|+|v_iz|$.
\end{Remark}

We next consider synchronizations with one single $\omega$-power.

\begin{Lemma}
  \label{l:conseq-FW}
  Let $\alpha=u_0(v_1) u_1\cdots (v_r) u_r$ be an $\omega$-term of rank 1, whose
  crucial portions are in normal form. Let $z$ be
  a Lyndon word, and let
  \begin{equation*}
    w=u_0^{}v_1^{n_1} u_1^{}\cdots v_r^{n_r}u_r^{}\in L_n[\alpha],
    \quad\text{with }n\geq\max\{\mu[\alpha],|z|+1\}.
  \end{equation*}
  Consider a prefix of $w$ of the form $pz^m$ with $m\geq n$ such
  that, for some $i\geq1$, the following inequalities hold:
  $$\bigl|w[i-1]|\leq|p|<|w[i]\bigr|.$$
  Then $z=v_i$ and
  \begin{enumerate}
  \item\label{item:conseq-FW-1a} either there is a factorization
    $u_{i-1}=qv_i^k$ such that $p=w[i-1]q$;
  \item\label{item:conseq-FW-1b} or there exists $k$ such that
    $p=w[i-1]u_{i-1}^{}v_i^k$.
  \end{enumerate}
\end{Lemma}

\begin{proof}
  Let $x=p^{-1}w[i]$. We claim that if $|x|\geq|v_iz|$, then $z^m$ and
  $v_i^{n_i}$ overlap on $|v_iz|$ positions. Suppose first that
  $|x|\geq|z^m|$. Then $z^m$ is a prefix of~$x$, which in turn is a
  suffix of $u_{i-1}^{}v_i^{n_i}$, so indeed $z^m$ and $v_i^{n_i}$
  overlap on $|z^m|-|u_{i-1}| \geq|v_iz|$ positions,
  by~\eqref{eq:overlap}. Consider next the case $|v_iz|\leq|x|<|z^m|$.
  Since
  $$\bigl|v_i^{n_i}\bigr|\geq\bigl|v_i^n\bigr|\geq\bigl|v_i^{|z|+1}\bigr|=(|z|+1)|v_i|\geq|v_iz|,$$
  one can consider the suffix $u$ of $v_i^{n_i}$ of length
  $|v_iz|$. Since $x$ and $v_i^{n_i}$ are suffixes of the same word
  and $|x| \geq|v_iz|$, $u$ is a suffix of $x$. Since $x$ is a prefix of
  $z^m$, $u$ is a factor of $z^m$. This proves the claim, so by
  Lemma~\ref{l:FW} applied to $v_i^{n_i}$ and $z^m$, we conclude that
  $z=v_i$, and that \ref{item:conseq-FW-1a} or~\ref{item:conseq-FW-1b}
  hold, depending on whether or not we
  have~$|p|<\big|w[i-1]u_{i-1}\big|$.

  Finally, assume that $|x|<|v_iz|$. From~\eqref{eq:overlap}, we get
  $|z^m|>|xu_i|$ and so $i<r$. Hence, using $m\geq\mu[\alpha]\geq2|v_iu_iv_{i+1}|
  \geq|v_iu_iv_{i+1}|+2$,
  \begin{equation*}
    |z^m|-|xu_i|
    \geq\bigl|z^{|v_iu_iv_{i+1}|+2}\bigr|-|xu_i|
    \geq|v_iu_iv_{i+1}|+2|z|-|xu_i|
    >|v_{i+1}z|.
  \end{equation*}
  Therefore, $z^m$ and $v_{i+1}^{n_{i+1}}$ have a common factor of
  length at least $|v_{i+1}z|$. By Lemma~\ref{l:FW} again, we have
  $z=v_{i+1}$ and $pz^k=w[i]u_i$ for some $k$ such that $1\leq k< m$.
  Since $|p|<\big|w[i]\big|$, it follows that $u_i$ is a suffix of
  $z^k=v_{i+1}^k$, contradicting the hypothesis that $v_i^\omega
  u_i^{}v_{i+1}^\omega$ is a crucial portion in normal form. This
  concludes the proof of the lemma.
\end{proof}

We now develop the inductive argument in order to prove
Proposition~\ref{p:synchro}.

\begin{Lemma}\label{l:conseq-FW-2}
  Let $\alpha=u_0(v_1) u_1\cdots (v_r) u_r$ and  $\beta=(z_1) y(z_2)$ be $\omega$-terms of rank 1 whose
  crucial portions are in normal form. Let
  \begin{equation*}
    w=u^{}_0v_1^{n_1} u^{}_1\cdots v_r^{n_r}u^{}_r\in L_n[\alpha], \text{\quad with }n\geq\max\{\mu[\alpha],\mu[\beta]\}.
  \end{equation*}
  If there is a prefix of $w$ of the form $pz_1^{m_1}yz_2^{m_2}$ with
  $m_1,m_2\geq n$ and
  \begin{equation}
    \label{eq:interval}
    \bigl|w[i-1]\bigr|\leq|p|<\bigl|w[i]\bigr|,
  \end{equation}
  then $z_1=v_i$, $i<r$, $y=u_i$, $z_2=v_{i+1}$, and $pz_1^{m_1}=w[i]$.
\end{Lemma}

\begin{proof}
  Lemma~\ref{l:conseq-FW} shows that $z_1=v_i$. We first assume
  $pz_1^{m_1}\preceq w[i]$. If $pz_1^{m_1}y\preceq w[i]$ (Case~$(a)$
  of Fig.~\ref{fig:l:conseq-FW-2-a}), then $y$ would be a prefix of a
  power of $v_i=z_1$, which is impossible since $\beta$ is in normal
  form. Hence $pz_1^{m_1}\preceq w[i]\prec pz_1^{m_1}y$, so $y$ and
  $u_i$ overlap. Consider the cases $pz_1^{m_1}y\preceq w[i]u_i$ and
  $pz_1^{m_1}y\succ w[i]u_i$ (Cases~$(b)$ and $(c)$ of
  Fig.~\ref{fig:l:conseq-FW-2-a}, in which the references to $v_{i+1}$
  underneath the straight line are justified below).
  \begin{figure}[ht]
    \centering
    \unitlength=.39mm
    \begin{picture}(150,60)(70,-45)
      \small
      \put(67.5,0){\line(1,0){157}}
      \curvel(80,-10,-20,130,-10,$z_1^{m_1}$,r)
      \curvel(130,-10,-20,160,0,$y$,r)
      \curvel(160,-10,-20,200,-10,$z_2^{m_2}$,r)
      \curvel(110,0,10,180,0,$v_i^{n_i}$,l)
      \straight(110,130,r,$v_i^{k}$,1)
      \node[Nframe=n](a)(140,-35){$(a)$}
    \end{picture}

    \begin{picture}(150,50)(70,-32)
      \small
      \put(65,0){\line(1,0){150}}
      \curvel(80,0,10,130,0,$v_i^{n_i}$,l)
      \curvel(130,0,10,180,0,$u_i$,l)
      \curvel(180,0,10,210,0,$v_{i+1}^{n_{i+1}}$,l)
      \curvel(70,-10,-20,110,-10,$z_1^{m_1}=v_i^{m_1}$,r)
      \curvel(110,-10,-20,150,0,$y$,r)
      \curvel(150,-10,-20,200,-10,$z_2^{m_2}=v_{i+1}^{m_2}$,r)
      \straight(110,130,r,$v_i^{k}$,1)
      \straight(130,150,r,\raise-2ex\hbox{$x$},1)
      \straight(150,180,r,$v_{i+1}^\ell$,1)
      \node[Nframe=n](b)(140,-35){$(b)$}
    \end{picture}
    \quad
    \begin{picture}(150,50)(70,-32)
      \small
      \put(67.5,0){\line(1,0){157}}
      \curvel(80,0,10,130,0,$v_i^{n_i}$,l)
      \curvel(130,0,10,160,0,$u_i$,l)
      \curvel(160,0,10,200,0,$v_{i+1}^{n_{i+1}}$,l)
      \curvel(70,-10,-20,110,-10,$z_1^{m_1}=v_i^{m_1}$,r)
      \curvel(110,-10,-25,180,0,$y$,r)
      \curvel(180,-10,-20,222,-10,$\,z_2^{m_2}=v_{i+1}^{m_2}$,r)
      \straight(110,130,r,$v_i^{k}$,1)
      \straight(160,180,r,$v_{i+1}^\ell$,1)
      \node[Nframe=n](c)(140,-35){$(c)$}
    \end{picture}
    \caption{Three factorization patterns when $pz_1^{m_1}\preceq w[i]$}
    \label{fig:l:conseq-FW-2-a}
  \end{figure}
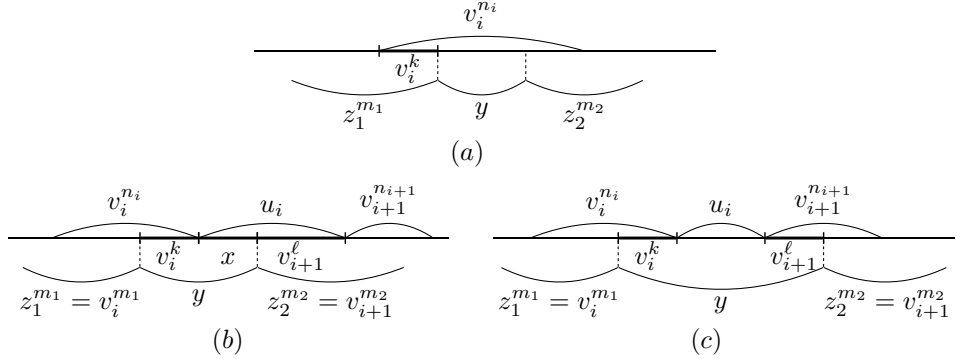
  We claim that $i>r$ and that $z_2^{m_2}$ and $v_{i+1}^{n_{i+1}}$
  overlap on $|v_{i+1}z_2|$ positions in~$w$.
  In Case~$(b)$, $z_2^{m_2}>|u_i|$ by~\eqref{eq:overlap} applied to
  $z=z_2$, hence $i<r$, and~\eqref{eq:overlap} applied at index $i+1$
  instead of $i$ yields $|z_2^{m_2}|\geq|u_i|+|v_{i+1}z_2|$, so
  $z_2^{m_2}$ and $v_{i+1}^{n_{i+1}}$ overlap on $|v_{i+1}z_2|$
  positions in~$w$. In Case~$(c)$, $i>r$ is clear from the assumption
  that $pz_1^{m_1}y\succ w[i]u_i$. Finally, we obtain
  $|v_{i+1}^{n_{i+1}}|=n_{i+1}|v_{i+1}|\geq\mu[\beta]|v_{i+1}|
  \geq2|z_1yz_2||v_{i+1}|$, whence
  $|v_{i+1}^{n_{i+1}}|-|y|\geq|z_2v_{i+1}|$. Similarly, we have
  $|z_2^{m_2}|\geq\mu[\alpha]|z_2| \geq|z_2v_{i+1}|$, so $z_2^{m_2}$
  and $v_{i+1}^{n_{i+1}}$ overlap on $|v_{i+1}z_2|$ positions.

  Therefore $z_2=v_{i+1}$ by Lemma~\ref{l:FW} and for some $k,\ell\geq0$, we
  have $y=v_i^kx$ and $u_i=xv_{i+1}^\ell$ in Case~$(b)$, where $x$ is the
  overlap between $y$ and $u_i$, and $y=v_i^ku_iv_{i+1}^\ell$ in
  Case~$(c)$. We claim that in either case, $k=\ell=0$, which proves the
  statement. In Case~$(b)$, $x$ is not a prefix of a power of $z_1$
  since, otherwise, so would be $y$, contradicting that $\beta$ is in
  normal form. On the other hand, $x$ is not a suffix of a power of
  $z_2$ since, otherwise, so would be~$u_i$, contradicting that
  $(v_i)u_i(v_{i+1})$ is in normal form. Therefore, $(v_i)x(v_{i+1})$
  is in normal form. Since $(v_i)u_i(v_{i+1})$ is also in normal form,
  we deduce in both cases, by condition~\ref{item:rank-i-4} of the
  definition of normal form, that $k=\ell=0$.

  \begin{figure}[ht]
    \centering
    \unitlength=.39mm
    \begin{picture}(150,60)(70,-45)
      \small
      \put(67.5,0){\line(1,0){157}}
      \curvel(80,-10,-20,130,-10,$v_i^{n_i}$,r)
      \curvel(130,-10,-20,160,0,$u_i$,r)
      \curvel(160,-10,-20,200,-10,$v_{i+1}^{n_{i+1}}$,r)
      \curvel(110,0,10,180,0,$z_1^{m_1}$,l)
      \straight(110,130,r,$v_i^{k}$,1)
      \node[Nframe=n](a)(140,-35){$(a)$}
    \end{picture}

    \begin{picture}(150,50)(70,-32)
      \small
      \put(65,0){\line(1,0){150}}
      \curvel(80,0,10,130,0,$z_1^{m_1}$,l)
      \curvel(130,0,10,180,0,$y$,l)
      \curvel(180,0,10,210,0,$z_2^{m_2}$,l)
      \curvel(70,-10,-20,110,-10,$v_i^{n_i}=z_1^{n_i}$,r)
      \curvel(110,-10,-20,150,0,$u_i$,r)
      \curvel(150,-10,-20,200,-10,$v_{i+1}^{n_{i+1}}=z_2^{n_{i+1}}$,r)
      \straight(110,130,r,$v_i^{k}$,1)
      \straight(150,180,r,$v_{i+1}^\ell$,1)
      \node[Nframe=n](b)(140,-35){$(b)$}
    \end{picture}
    \quad
    \begin{picture}(150,50)(70,-32)
      \small
      \put(67.5,0){\line(1,0){157}}
      \curvel(80,0,10,130,0,$z_1^{m_1}$,l)
      \curvel(130,0,10,160,0,$y$,l)
      \curvel(160,0,10,200,0,$v_{i+1}^{n_{i+1}}$,l)
      \curvel(70,-10,-20,110,-10,$v_i^{n_i}=z_1^{m_1}$,r)
      \curvel(110,-10,-25,180,0,$u_i$,r)
      \curvel(180,-10,-20,222,-10,$v_{i+1}^{m_2}$,r)
      \straight(110,130,r,$v_i^{k}$,1)
      \straight(160,180,r,$v_{i+1}^\ell$,1)
      \node[Nframe=n](c)(140,-35){$(c)$}
    \end{picture}
    \caption{Three factorization patterns when $pz_1^{m_1}\succeq w[i]$}
    \label{fig:l:conseq-FW-2-b}
  \end{figure}
  Finally, assume that $pz_1^{m_1}\succeq w[i]$. The resulting three
  factorization patterns are depicted in
  Fig.~\ref{fig:l:conseq-FW-2-b}. Note that they are in correspondence
  with the factorization patterns in Fig.~\ref{fig:l:conseq-FW-2-a}.
  The arguments presented above for the case $pz_1^{m_1}\preceq w[i]$
  therefore apply, \emph{mutatis mutandis}, to the current case.
\end{proof}

\begin{proof}[Proof of Proposition~\ref{p:synchro}]
  We have $|v_1^{n_1}|\geq|v_1|+n_1-1\geq|v_1|+2|z_0t_1|-1\geq|z_0t_1v_1|$.
  Likewise, $|t_1^{m_1}|\geq|u_0v_1t_1|$, so $v_1^{n_1}$ and $t_1^{m_1}$ overlap
  on a factor of length at least $|v_1t_1|$. By
  Lemma~\ref{l:FW}, $v_1=t_1$ and $z_0=u_0t_1^k$ or $u_0=z_0v_1^k$ for some
  $k$. Since $\alpha$ and $\beta$ are in normal form, $k=0$ by
  property~\ref{item:rank-i-4} of normal forms. Hence $u_0=z_0$.
  Suppose inductively that for $i\geq1$, we have $u_{k-1}=z_{k-1}$,
  $n_{k-1}=m_{k-1}$ and $v_k=t_k$ for all $1\leq k\leq i$.
  If $i<s$, then one can apply Lemma~\ref{l:conseq-FW-2}, since the
  portion $(t_i)z_i(t_{i+1})$ is in normal form and
  $w[i-1]\preceq z_0t_1^{m_1}z_1\cdots t_{i-1}^{m_{i-1}}z_{i-1}\prec w[i]$, where the
  notation $w[\cdot]$ refers to the first factorization, so
  that~\eqref{eq:interval} is fulfilled for the word
  $p=z_0t_1^{m_1}z_1\cdots t_{i-1}^{m_i-1}z_{i-1}$. This yields $i<r$,
  $u_i=z_i$, $n_i=m_i$ and $v_{i+1}=t_{i+1}$. The case $i<r$ is
  dual. Finally, if $i=r=s$, then we obtain
  $u_r=z_r$ by left-right symmetry, so $n_r=m_r$.
\end{proof}

\section{\texorpdfstring{The $\omega$-word problem over \pv A}{The omega-terms-word problem over A}}
\label{sec:omega-word-problem-over-A}

In this section we reveal how the languages $L_n[\alpha]$ can be used to
obtain an alternative proof of McCammond's solution of the word
problem for $\omega$-terms over \pv A. The fundamental property of the
languages $L_n[\alpha]$, whose proof is presented in the next section, is
their star-freeness under suitable hypotheses.

\begin{Thm}
  \label{t:star-free}
  Let $\alpha$ be a normal form $\omega$-term and let $n\geq\mu[\alpha]$. Then the
  language $L_n[\alpha]$ is star-free.
\end{Thm}

A simpler but also important property is stated in
Lemma~\ref{l:separation} below, which follows from the synchronization
property of Proposition~\ref{p:synchro}.  For an $\omega$-term $\alpha$,
we set $E_n^*[\alpha]=\bigcup_{i\geq0}E_n^i[\alpha]$.  

\begin{Lemma}
  \label{l:separation}
  Let $\alpha$ and $\beta$ be two $\omega$-terms in normal form with $\rk \beta\geq \rk \alpha$,
  and let $n>\max\{\mu[\alpha],\mu[\beta]\}$. If $L_n[\alpha]\cap L_n[\beta]\neq\emptyset$, then $\alpha\in E_n^*[\beta]$.
\end{Lemma}

\begin{proof}
  Let $w\in L_n[\alpha]\cap L_n[\beta]$. We proceed by induction on $\rk \alpha=i$. If
  $i=0$, that is $\alpha\in X^+$, we have $w=\alpha$ so that $\alpha\in L_n[\beta]=E_n^{\rk
    \beta}[\beta]$. Assume next $i\geq 1$ and that the result holds for $\rk
  \alpha<i$. By definition of $L_n$ and the choice of~$w$, there exist
  $\alpha_1\in E_n^{\rk \alpha-1}[\alpha]$ and $\beta_1\in E_n^{\rk \beta-1}[\beta]$ such that $w\in
  L_n[\alpha_1]\cap L_n[\beta_1]$. By Lemma~\ref{l:McCammond-Lemma-10.7}, the
  $\omega$-terms $\alpha_1$ and $\beta_1$ are in normal form. Let $u_0(v_1) u_1\cdots
  (v_r) u_r$ and $z_0(t_1) z_1\cdots (t_s) z_s$ be the normal form
  expressions of $\alpha_1$ and $\beta_1$, respectively. We have
  $n\geq\max\{\mu[\alpha_1],\mu[\beta_1]\}$ by Lemma~\ref{l:mu-vs-expansions}. By
  Proposition~\ref{p:synchro}, it follows that $\alpha_1=\beta_1$, so
  \begin{equation}
    \label{eq:separation-1}
    E_n^{\rk \alpha-1}[\alpha]\cap E_n^{\rk \beta-1}[\beta]\neq\emptyset.
  \end{equation}
  If $i=1$, then $\alpha_1=\alpha$ so that $\alpha\in E_n^{\rk
    \beta-1}[\beta]$. If $i>1$, consider the freezes $\freeze \alpha$
  and $\freeze \beta$. Then $L_n[\freeze \alpha]\cap L_n[\freeze
  \beta]\neq\emptyset$ follows from~(\ref{eq:separation-1}), and by
  Fact~\ref{remark:freezed-parentheses}, $\freeze \alpha$ and $\freeze
  \beta$ are in normal form, $n>\max\{\mu[\freeze \alpha],\mu[\freeze
  \beta]\}$, and $\rk {\freeze \beta}\geq \rk{\freeze \alpha}=i-1$. By
  the induction hypothesis, we obtain $\freeze \alpha\in E_n^*[\freeze
  \beta]$ and, therefore, $\alpha\in E_n^*[\beta]$, which completes
  the induction step and the proof of the lemma.
\end{proof}

By raising the lower bound for $n$, we obtain a more precise result.

\begin{Thm}\label{t:separation}
  Let $\alpha$ and $\beta$ be two $\omega$-terms in normal form and
  let $n$~be an integer such that
  $n>\max\{|\alpha|,|\beta|,\mu[\alpha],\mu[\beta]\}$. If
  $L_n[\alpha]\cap L_n[\beta]\neq\emptyset$, then $\alpha=\beta$.
\end{Thm}

\begin{proof}
  Suppose that $\rk \alpha\leq \rk \beta$, so
  that, by Lemma~\ref{l:separation}, $\alpha\in E_n^*[\beta]$. If $\rk \beta>\rk \alpha$, it follows that $|\alpha|\geq
  n$, which contradicts the assumption on $n$. Hence we must have $\rk
  \beta=\rk \alpha$ and so $\alpha=\beta$.
\end{proof}

Combining Theorems~\ref{t:star-free} and~\ref{t:separation}, we obtain
a new proof of uniqueness of McCammond's normal form for elements
of~\omc XA.

\begin{Cor}[McCammond's solution of the $\omega$-word problem over \pv
  A~\cite{McCammond:1999a}]
  \label{c:word-problem}
  Let $\alpha$ and $\beta$ be $\omega$-terms in normal form which define the
  same pseudoword over~\pv A, that is, such that $\epsilon[\alpha]=\epsilon[\beta]$. Then $\alpha=\beta$.
\end{Cor}

\begin{proof}
  Let $n>\max\{|\alpha|,|\beta|,\mu[\alpha],\mu[\beta]\}$. Since
  $L_n[\alpha]$ and $L_n[\beta]$ are star-free languages by
  Theorem~\ref{t:star-free}, their respective closures $\closv
  A{L_n[\alpha]}$ and $\closv A{L_n[\beta]}$ in~\Om XA are clopen
  subsets. Since $\epsilon[\alpha]=\epsilon[\beta]\in\closv
  A{L_n[\alpha]}\cap\closv A{L_n[\beta]}$, the nonempty open set
  $\closv A{L_n[\alpha]}\cap\closv A{L_n[\beta]}$ contains some
  elements of the dense set~$X^+$, which in turn belong to
  $L_n[\alpha]\cap L_n[\beta]$ since,
  by~\cite[Theorem~3.6]{Almeida:2003c}, we have $\closv
  A{L_n[\gamma]}\cap X^+=L_n[\gamma]$ ($\gamma\in\{\alpha,\beta\}$).
  Therefore, $L_n[\alpha]\cap L_n[\beta]\neq\emptyset$, whence
  $\alpha=\beta$ by Theorem~\ref{t:separation}.
\end{proof}

\section{\texorpdfstring{Star-freeness of the languages
    $L_n[\alpha]$}{Star-freeness of the languages
    Ln[alpha]}}
\label{sec:star-freeness-Ln}

This section is dedicated to the proof of Theorem~\ref{t:star-free}.

We say that an $\omega$-term $\alpha$ is in \emph{circular normal form} if the
crucial portions of $\alpha^2$ are in normal form.
A consequence of Lemma~\ref{l:McCammond-Lemma-10.7}, is that the
property of being in circular normal form is preserved by expansions.

\begin{Lemma}
  \label{l:cnf-expansion}
  Let $\alpha$ be an $\omega$-term in circular normal form and let $\beta\in
  E_n[\alpha]$. Then $\beta$ is also in circular normal form.
\end{Lemma}

\begin{proof}
  If
  $\alpha$ is a word, then $\beta=\alpha$ is certainly in circular normal
  form. Otherwise $\alpha$ and $\beta$ are of the form $\alpha=\gamma_0(\delta_1)\gamma_1\cdots(\delta_r)\gamma_r$
  and $\beta=\gamma_0\delta_1^{n_1}\gamma_1\cdots\delta_r^{n_r}\gamma_r$ with each $n_k\geq n$. Now, $\beta^2\in
  E_n[\alpha]^2=E_n[\alpha^2]$ according to Lemma~\ref{l:Ln}~\ref{item:Ln-1}. By
  Lemma~\ref{l:McCammond-Lemma-10.7}, applied to the crucial portions
  of~$\alpha^2$ which are by hypothesis in normal form, we conclude that
  all factors $\delta_k^{n_k}\gamma_k\delta_{k+1}^{n_{k+1}}$, as well as
  $\delta_r^{n_r}\gamma_r\gamma_0\delta_{1}^{n_1}$, are in normal form.  Since each
  crucial portion of~$\beta^2$ is a crucial portion of one of these
  factors, it is in normal form, hence $\beta$ is in circular normal~form.
\end{proof}

Let us now derive a corollary of Proposition~\ref{p:synchro}, which applies to
$\omega$-terms in circular normal form (rather than to $\omega$-terms in normal form as
in the proposition).

\begin{Cor}[of Prop.~\ref{p:synchro}]\label{c:synchro}
Let $\alpha=(v_1) u_1\cdots (v_r) u_r$ and $\beta=(t_1) z_1\cdots (t_s) z_s$ be  two $\omega$-terms of rank~1 in circular normal form, and let $n\geq\max\{\mu[\alpha],\mu[\beta]\}$. If $L_n[\alpha]\cap L_n[\beta]\neq\emptyset$, then $\alpha=\beta$.
\end{Cor}

\begin{proof}
  Let $w=v_1^{n_1} u_1\cdots v_r^{n_r} u_r=t_1^{m_1} z_1\cdots t_s^{m_s}
  z_s\in L_n[\alpha]\cap L_n[\beta]$ and let $\alpha'$ and $\beta'$ be
  respectively the normal forms of $\alpha$ and $\beta$.  As $\alpha$ is in
  circular normal form by hypothesis, all its crucial portions are in normal
  form. Therefore $\alpha'$ is obtained from $\alpha$ by simply reducing
  the final portion $(v_r) u_r$ to its normal form. This is done by 
  applying all possible, say $k\geq 0$, reductions of type~$4R$. That is,
  $\alpha'=(v_1) u_1\cdots (v_r) u'_r$ with $u_r=v_r^ku'_r$. Analogously,
  $\beta'=(t_1) z_1\cdots (t_s) z'_s$ with $z_s=t_s^\ell z'_s$ for some
  $\ell\geq0$. Clearly $\mu[\alpha]\geq \mu[\alpha']$ and $\mu[\beta]\geq
  \mu[\beta']$, whence $n\geq\max\{\mu[\alpha'],\mu[\beta']\}$. On the other
  hand, $w=v_1^{n_1} u_1\cdots v_r^{n_r+k} u'_r=t_1^{m_1} z_1\cdots
  t_s^{m_s+\ell} z'_s$ belongs to $L_n[\alpha']\cap L_n[\beta']$. Hence,
  $\alpha'=\beta'$ by Proposition~\ref{p:synchro}. In particular $v_1=t_1$,
  $v_r=t_s$ and $u'_r=z'_s$.  The crucial portions $(v_r) u_r(v_1)$ and $(t_s)
  z_s(t_1)$ of, respectively, $\alpha^2$ and $\beta^2$ are in normal
  form. Then, as $(v_r) u_r(v_1)=(v_r) v_r^ku'_r(v_1)$ and $(t_s)
  z_s(t_1)=(v_r) v_r^\ell u'_r(v_1)$, we deduce from
  property~\ref{item:rank-i-4} of normal forms that $k=\ell$. This completes
  the proof that $\alpha=\beta$.
\end{proof}

The next lemma reflects periodicities of sufficiently large expansions
of an $\omega$-term of rank~1 in the term itself, provided it is in circular
normal form.

\begin{Lemma}
  \label{l:base-primitive-rank-1}
  Let $\alpha$ be an $\omega$-term of rank 1 in circular normal form and let
  $n\geq\mu[\alpha]$. If $z^\ell\in L_n[\alpha]$, then there exists an $\omega$-term of rank~1
  in circular normal form $\zeta$ such that $\alpha=\zeta^\ell$ and $z\in L_n[\zeta]$.
\end{Lemma}

\begin{proof}
  Let $\alpha=u_0(v_1)u_1\cdots(v_r)u_r$. Since $z^\ell\in L_n[\alpha]$, either $z\prec u_0\prec
  z^\ell$ or $u_0\preceq z$. In both cases, we reduce the question to the case
  where $u_0$ is empty, by replacing $\alpha$ by $(v_1)u_1\cdots(v_r)u_ru_0$ and
  $z$ by an appropriate conjugate: $z_1^{-1}zz_1$, where $z=z_1z_2$
  and $u_0=(z_1z_2)^kz_1$ in the first case, or $u_0^{-1}zu_0$ in the
  second case. So write $\alpha=(v_1)u_1\cdots(v_r)u_r$, and let $w=z^\ell$. Then,
  we have $w=v_1^{n_1}u_1\cdots v_r^{n_r}u_r$, with $n_1,\ldots, n_r\geq n$.
  Taking into account the resulting factorization of~$w^2$, we also
  set $v_{r+i}=v_i$, $u_{r+i}=u_i$, and $n_{r+i}=n_i$ for
  $i=1,\ldots,r$. Note also that $w^2[i+r]=ww[i]$ for $1\leq i\leq r$.

  If $\ell=1$, then we choose $\zeta=\alpha$. For $\ell\geq2$,
  assume first that $|z|\leq|v_1^{n_1}|$. Since both $z$ and
  $v_1^{n_1}$ are prefixes of $w$, this implies that $v_1^{k-1}\prec
  z\preceq v_1^k$ for some $k\geq1$. Then $t=z^{-1}v_1^k$ is a suffix
  of $(v_1^{k-1})^{-1}v_1^k=v_1$. Further, $v_1\preceq w$ and
  $t\preceq z^{-1}w=z^{\ell-1}\prec z^\ell=w$, so $t\preceq v_1$.
  Since $t$ is both a prefix and a suffix of the Lyndon word $v_1$,
  $t$ is either empty or equal to $v_1$ by
  Lemma~\ref{l:Lyndon-less-than-suffices}, hence $z=v_1^k$ and
  $w=z^\ell=v_1^{k\ell}$. It follows that $u_1\preceq
  v_1^{k\ell-n_1}$, which contradicts the hypothesis on~$\alpha$.
  Therefore, we have $|v_1^{n_1}|<|z|$ and $v_1^{n_1}\prec z$.

  In particular, the equalities $|z|>n_1\geq n\geq\mu[\alpha]>|u_r|$
  hold, so that $w[1]\prec z\preceq z^{\ell-1}\prec w[r]$. Hence $r\geq2$
  and there exists $i\in\{1,\ldots,r-1\}$ such that $w[i] \preceq z
  \prec w[i+1]$, which is the same as $w^2[i] \preceq z \prec
  w^2[i+1]$. We prove the following property by induction on
  $k\in\{1,\ldots,r\}$:
  \begin{equation}
    \label{eq:z-induction}
    \left\{
      \begin{aligned}
        &zw[j]=w^2[i+j]\text{ and } u_j=u_{i+j}\\
        &v_j=v_{i+j}\end{aligned}
    \right.
    \begin{aligned}
      &\qquad\text{for $j\leq k$},\\
      &\qquad\text{for $j\leq k+1$}.
    \end{aligned}
    \tag*{$H(k)$}
  \end{equation}
  Observe that $zwv_1^{n_1}\preceq z^{\ell+2}\preceq z^{2\ell}=w^2$.
  We will apply several times Lemma~\ref{l:conseq-FW-2} to $\alpha^2$,
  choosing prefixes of $zwv_1^{n_1}$ for the successive values of the
  prefix $p$ of $w^2\in L_n[\alpha^2]$ which is considered in that
  lemma. First, since $w^2[i] \preceq z \prec w^2[i+1]$ and
  $n\geq\mu[\alpha]=\mu[\alpha^2]$, we may apply
  Lemma~\ref{l:conseq-FW-2} to~$\alpha^2$, with $\beta=(v_1)u_1(v_2)$
  and $p=z$, to obtain $zw[1]=zv_1^{n_1}=w^2[i+1]$, $v_1=v_{i+1}$,
  $u_1=u_{i+1}$, and $v_2=v_{i+2}$, which establishes~$H(1)$. Next,
  assuming that $H(k-1)$ holds for a certain $k\leq r$, we deduce that
  $w^2[i+k-1]\preceq zw[k-1]u_{k-1}\prec w^2[i+k]$.
  Lemma~\ref{l:conseq-FW-2} applied to~$\alpha^2$ with
  $\beta=(v_k)u_k(v_{k+1})$ and $p=zw[k-1]u_{k-1}$ then yields~$H(k)$.

  In particular $zw[r]=w^2[i+r]=ww[i]$ and $u_r=u_{i+r}=u_i$. It
  follows that $zw=zw[r]u_r=ww[i]u_i$. Since $zw=wz\,(=z^{\ell+1})$,
  we deduce that $z=w[i]u_i=v_{1}^{n_{1}}u_{1}\cdots v_i^{n_i}u_i$.
  Let $\zeta=(v_1)u_1\cdots(v_i)u_i$. Then $z$ belongs to
  $L_n[\zeta]$ and $z^\ell\in L_n[\zeta^\ell]\cap
  L_n[\alpha]$. Since each crucial portion of $\zeta^2$ is a crucial portion of
  $\alpha^2$, we have $\mu[\zeta^\ell]=\mu[\zeta]\leq\mu[\alpha]\leq n$. Therefore, $\alpha=\zeta^\ell$ by
  Corollary~\ref{c:synchro}, which completes the proof.
\end{proof}

We call \emph{primitive} an $\omega$-term which is primitive when
represented as a parenthesized word. An immediate consequence of
Lemma~\ref{l:base-primitive-rank-1} is the following observation.

\begin{Cor}
  \label{c:base-primitive-rank-1}
  Let $\alpha$ be an $\omega$-term of rank 1 in circular normal form
  and let $n\geq\mu[\alpha]$. If $\alpha$ is primitive and $w\in
  L_n[\alpha]$, then $w$ is also primitive.\qed
\end{Cor}

The next result may be regarded as a generalization of
Corollary~\ref{c:base-primitive-rank-1} to $\omega$-terms of larger
rank.

\begin{Cor}
  \label{c:base-primitive-rank-i}
  Let $\alpha$ be an $\omega$-term of rank $i\geq1$ in circular normal
  form and let $n\geq\mu[\alpha]$. If $\alpha$ is primitive and
  $\beta\in E_n[\alpha]$, then $\beta$ is also primitive.
\end{Cor}

\begin{proof}
  We distinguish two types of parentheses in the $\omega$-term $\alpha$: write
  $(,)$ for the parentheses corresponding to the $\omega$-powers of
  rank~$i$, and ${\op},{\cl}$ for the remaining parentheses. Consider
  the alphabet $Z=X\cup\{{\op},{\cl}\}$, with the extended ordering
  ${\op}< x<{\cl}$ ($x\in X$). Then $\beta$ may be viewed as a word $\beta_Z$
  over $Z$ and $\alpha$ as an $\omega$-term $\alpha_Z$, of rank~1, over the same
  alphabet such that $\beta_Z\in L_n[\alpha_Z]$. Moreover
  $\mu[\alpha_Z]\leq\mu[\alpha]$ and it is
  clear by McCammond's definition of rank~$i$ normal form that $\alpha_Z$
  is a primitive $\omega$-term in circular normal form (over $Z$), whence
  $\alpha_Z$ and $\beta_Z$ satisfy the hypotheses of
  Corollary~\ref{c:base-primitive-rank-1}. To conclude the proof, it
  suffices to invoke Corollary~\ref{c:base-primitive-rank-1}.
\end{proof}

Iterating the application of Corollary~\ref{c:base-primitive-rank-i},
we obtain another extension of
Corollary~\ref{c:base-primitive-rank-1} to $\omega$-terms of any rank.

\begin{Prop}
  \label{p:base-primitive}
  Let $\alpha$ be an $\omega$-term in circular normal form and let $n\geq\mu[\alpha]$. If
  $\alpha$ is a primitive $\omega$-term and $w\in L_n[\alpha]$, then $w$ is a primitive
  word.
\end{Prop}

\begin{proof}
  We proceed by induction on~$\rk \alpha$. The case $\rk \alpha=1$ is given by
  Corollary~\ref{c:base-primitive-rank-1}. Assume that the result
  holds for $\omega$-terms whose rank is~$\rk \alpha-1\geq1$. By definition of $L_n[\alpha]$,
  there is an $\omega$-term $\alpha'\in E_n[\alpha]$ such that $w\in L_n[\alpha']$. By
  Corollary~\ref{c:base-primitive-rank-i}, $\alpha'$ is
  primitive. Moreover, $\mu[\alpha']\leq\mu[\alpha]$ by Lemma~\ref{l:mu-vs-expansions},
  and $\alpha'$ is in circular normal form by
  Lemma~\ref{l:cnf-expansion}. Hence, by induction hypothesis,
  $w$ is primitive, which completes the induction step.
\end{proof}

The following result generalizes Lemma~\ref{l:base-primitive-rank-1}
in case $\alpha$ is a primitive $\omega$-term.

\begin{Lemma}
  \label{l:root}
  Let $\alpha$ be a primitive $\omega$-term in circular normal form
  and let $n\geq\mu[\alpha]$. If $z^\ell\in L_n[\alpha]^k$ then $z\in
  L_n[\alpha]^m$ for some $m$ such that $1\leq m\leq k$.
\end{Lemma}

\begin{proof}
  We proceed by induction on~$\rk \alpha$. If $\rk \alpha=0$, then
  $\alpha$ is a word, so $L_n[\alpha]=\{\alpha\}$ and
  $z^\ell=\alpha^k$. By~\cite[Prop.~1.3.1]{Lothaire:1983}, $z$ and
  $\alpha$ are powers of the same word, whence $z=\alpha^m$ since
  $\alpha$ is primitive, and $m=k/\ell\leq k$. Assume now that $\rk
  \alpha\geq1$ and that the result holds at lower ranks.

  Since $L_n[\alpha]^k=L_n[\alpha^k]$ by
  Lemma~\ref{l:Ln}\ref{item:Ln-2}, we have $z^\ell\in E_n[E_n^{\rk
    \alpha-1}[\alpha^k]]$. Pick an $\omega$-term $\beta\in
  E_n^{\rk\alpha-1}[\alpha^k]$ of rank~$1$ such that $z^\ell\in
  E_n[\beta]=L_n[\beta]$. Since $\alpha$ is in circular normal form,
  so is $\alpha^k$, whence so is $\beta$ by
  Lemma~\ref{l:cnf-expansion}. Since
  $n\geq\mu[\alpha]=\mu[\alpha^k]\geq\mu[\beta]$ by
  Lemma~\ref{l:mu-vs-expansions}, one can apply
  Lemma~\ref{l:base-primitive-rank-1}: there exists an $\omega$-term
  $\zeta$ of rank~$1$ such that $\beta=\zeta^\ell$ and $z\in
  L_n[\zeta]$.

  If $\rk \alpha=1$, then $\zeta^\ell\in E_n^{\rk
    \alpha-1}[\alpha^k]=\{\alpha^k\}$. Since $\alpha$ is primitive, it
  follows that $\zeta=\alpha^m$, for some~$m\leq k$, and $z\in
  L_n[\zeta]=L_n[\alpha]^m$, as required. If $\rk \alpha>1$, let us
  check that we may apply the induction hypothesis to the freeze
  $\freeze \alpha$ of $\alpha$ and $\freeze\zeta\in(X\cup\{{\op},{\cl}\})^*$.
  First, $\freeze \zeta^\ell\in L_n[\freeze \alpha^k]$. Next, since
  the crucial portions of
  $\alpha^2$ are in normal form, so are those of
  $\freeze{\alpha\cdot\alpha}=\freeze \alpha\cdot\freeze \alpha$ by
  Fact~\ref{remark:freezed-parentheses}, whence $\freeze \alpha$ is in
  circular normal form. Finally, the relations
  $n\geq\mu[\alpha]\geq\mu[\freeze \alpha]$ and $\rk{\freeze
    \alpha}=\rk \alpha-1$ hold, also by
  Fact~\ref{remark:freezed-parentheses}. By induction, we obtain
  therefore $m$ such that $1\leq m\leq k$ and $\freeze \zeta\in
  L_n[\freeze \alpha^m]$. We deduce that $\zeta\in E_n^{\rk
    \alpha-1}[\alpha^m]$, so that $z\in L_n[\zeta]\subseteq E_n^{\rk
    \alpha}[\alpha^m]=L_n[\alpha^m]$.
\end{proof}

We proceed to establish the following important property of the
languages $L_n[\alpha]$ for primitive $\omega$-terms $\alpha$. In its
proof, we apply in both directions Sch\"utzenberger's Theorem
\cite{Schutzenberger:1965}, stating that a language is star-free if
and only if its syntactic semigroup is finite and satisfies the
pseudoidentity $x^{\omega+1}=x^\omega$.

\begin{Lemma}
  \label{l:pump-bases}
  Let $\alpha$ be a primitive $\omega$-term in circular normal form
  and let $n\geq\mu[\alpha]$. If $L_n[\alpha]$ is a star-free
  language, then so is $L_n[\alpha]^*$.
\end{Lemma}

\begin{proof}
  Let $M$ be an integer such that the syntactic semigroup
  of~$L_n[\alpha]$ satisfies the identity $x^M=x^{M+1}$ and let $K$ be
  a positive integer to be identified later. Let $N>MK$ be an integer
  and suppose that $x,y,z$ are words such that $xy^Nz\in
  L_n[\alpha]^*$. The result follows from the claim that, for
  sufficiently large $K$, depending only on $\alpha$ and $n$,
  $xy^{N+1}z$ belongs to~$L_n[\alpha]^*$.

  To prove the claim, we start with a factorization $xy^Nz=w_1\cdots
  w_m$ where each $w_j\in L_n[\alpha]$. Consider each product of $M$
  consecutive $y$'s within the factor $y^N$. If at least one of the
  factors appears completely within one of the $w_j$, then we have a
  factorization $w_j=x'y^Mz'$ as indicated in
  Figure~\ref{fig:l:pump-bases-case-short}.
  \begin{figure}[ht]
    \centering
    \unitlength=.45\unitlength
    \begin{picture}(250,35)(0,-20)
      \small
      \put(0,0){\line(1,0){250}}
      \put(0,-2){\line(0,1){4}}
      \put(250,-2){\line(0,1){4}}
      \curvel(0,0,10,23,0,$w_1$,l)
      \put(28,0){$\cdots$}
      \curvel(40,0,10,70,0,$w_{j-1}$,l)
      \curvel(70,0,10,160,0,$w_j$,l)
      \curvel(160,0,10,200,0,$w_{j+1}$,l)
      \put(208,0){$\cdots$}
      \curvel(222,0,10,250,0,$w_m$,l)

      \curvel(90,0,-10,130,0,$y^M$,r)

      \curvel(0,-8,-18,90,0,$xy^p$,r)
      \curvel(130,-8,-18,250,0,$y^{N-M-p}z$,r)

      \straight(70,90,r,$x'$,.7)
      \straight(130,160,r,$z'$,.7)
    \end{picture}
    \caption{Case where some $y^M$ falls within some $w_j$}
    \label{fig:l:pump-bases-case-short}
  \end{figure}
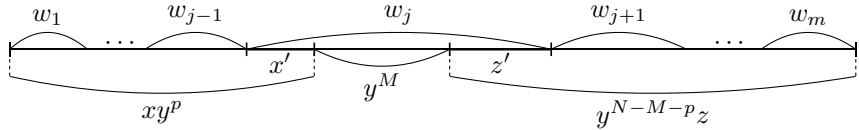
  In particular, the word $x'y^Mz'$ belongs to the star-free language
  $L_n[\alpha]$. By the choice of~$M$, we deduce that
  $w_j'=x'y^{M+1}z'\in L_n[\alpha]$. Hence, for $p$ as in
  Figure~\ref{fig:l:pump-bases-case-short},
  $$xy^{N+1}z
  =xy^p\cdot y^{M+1}\cdot y^{N-M-p}z =w_1\cdots
  w_{j-1}w_j'w_{j+1}\cdots w_m$$
  is again a word from $L_n[\alpha]^m$, independently of the value
  of~$K\geq1$.

  We may therefore assume that no factor $y^M$ appears completely
  within some factor $w_j$. Thus, each of the
  first $K<N/M$ consecutive factors $y^M$, which form a prefix
  of~$y^N$, as well as the product $y^{N-KM}z$, start in a different
  $w_j$, say in $w_{j_1},\ldots,w_{j_{K+1}}$, with
  $j_1<\cdots<j_{K+1}$. This determines factorizations
  \begin{align}
    w_{j_s}&=w_{j_s,1}w_{j_s,2}
    \label{eq:pump-bases-1}\\
    y^M&=w_{j_s,2}x_sw_{j_{s+1},1}\ (s=1,\ldots,K)
    \label{eq:pump-bases-2}\\
    x&=x'w_{j_1,1}
    \nonumber\\
    y^{N-KM}z&=w_{j_{K+1},2}z'
    \nonumber
  \end{align}
  where each $x_s$, $x'$, and $z'$ is a word from $L_n[\alpha]^*$, as
  represented in Figure~\ref{fig:l:pump-bases-case-long}.
  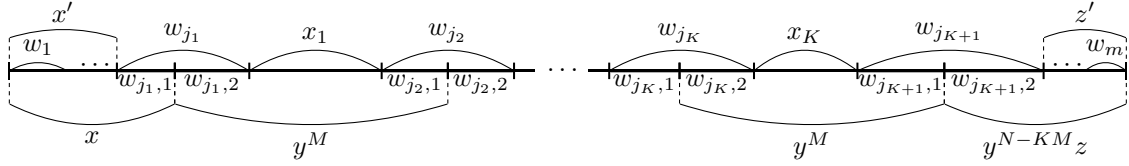
\begin{figure}[h]
    \centering
    \unitlength=.55\unitlength
    \begin{picture}(270,40)(0,-20)
      \small
      \put(0,0){\line(1,0){127}}
      \put(130,0){$\ldots$}
      \put(140,0){\line(1,0){130}}
      \put(0,-2){\line(0,1){4}}
      \put(270,-2){\line(0,1){4}}
      \curvel(0,0,4,14,0,$w_1$,l)
      \put(17,0){$\cdots$}
      \curvel(0,8,12,26,0,$x'$,l)
      \curvel(26,0,10,58,0,$w_{j_1}$,l)
      \curvel(58,0,10,90,0,$x_1$,l)
      \curvel(90,0,10,122,0,$w_{j_2}$,l)

      \curvel(145,0,10,180,0,$w_{j_{K}}$,l)
      \curvel(180,0,10,205,0,$x_{K}$,l)
      \curvel(205,0,10,250,0,$w_{j_{K+1}}$,l)
      \put(252,0){$\cdots$}
      \curvel(260,0,4,270,0,$w_m$,l)
      \curvel(250,8,12,270,0,$z'$,l)

      \curvel(0,-8,-18,40,0,$x$,r)
      \curvel(40,-8,-18,106,0,$y^M$,r)

      \curvel(162,-8,-18,226,0,$y^M$,r)
      \curvel(226,-8,-18,270,0,$y^{N-KM}z$,r)

      \straight(26,40,r,$w_{j_1,1}$,.7)
      \straight(40,58,r,$w_{j_1,2}$,.7)
      \straight(90,106,r,$w_{j_2,1}$,.7)
      \straight(106,122,r,$w_{j_2,2}$,.7)

      \straight(145,162,r,$w_{j_{K},1}$,.7)
      \straight(162,180,r,$w_{j_{K},2}$,.7)
      \straight(205,226,r,$w_{j_{K+1},1}$,.7)
      \straight(226,250,r,$w_{j_{K+1},2}$,.7)
    \end{picture}
    \caption{Case where each $y^M$ overlaps several $w_j$}
    \label{fig:l:pump-bases-case-long}
  \end{figure}

  Consider a finite deterministic automaton recognizing the language
  $L_n[\alpha]$. Each pair of words $(w_{j_s,1},w_{j_s,2})$ determines
  two consecutive paths leading from the initial state to a final
  state. Thus, if $K$~is greater than the number of states, then there
  exist two indices $p,q$ such that $1\leq p<q\leq K$ and the words
  $w_{j_p,1}$ and $w_{j_q,1}$ both lead from the initial state to the
  same state. It follows that $w_{j_q,1}w_{j_p,2}$ belongs to
  $L_n[\alpha]$.
  Hence the word
  \begin{align*}
    w_{j_{p+1},1}y^{M(q-p-2)}w_{j_{q-1},2}x_{q-1}w_{j_q,1}\cdot w_{j_p,2}x_p
    &=w_{j_{p+1},1}y^{M(q-p-1)}w_{j_p,2}x_p\\
    &=(w_{j_{p+1},1}w_{j_p,2}x_p)^{q-p}
  \end{align*}
  belongs to~$L_n[\alpha]^*$ where, for the second equality, we use the
  factorization~\eqref{eq:pump-bases-2} with $s=p$ for each $y^M$.
  Now, $w_{j_{p+1},1}w_{j_p,2}x_p$ is a conjugate of $y^M$ again by
  \eqref{eq:pump-bases-2} and, therefore, it is of the form $t^M$,
  where $t$ is a conjugate of~$y$. By Lemma~\ref{l:root}, $t$ belongs
  to~$L_n[\alpha]^*$. On the other hand, note that
  \begin{align*}
    xy^Nz
    &=x'w_{j_1}x_1\cdots w_{j_{p-1}}x_{p-1}w_{j_p}x_p
    \cdot
    w_{j_{p+1}}x_{p+1}\cdots w_{j_K}x_Kw_{j_{K+1}}z'
    \\
    xy^{N+1}z
    &=x'w_{j_1}x_1\cdots w_{j_{p-1}}x_{p-1}w_{j_p}x_p
    \cdot
    t
    \cdot
    w_{j_{p+1}}x_{p+1}\cdots w_{j_K}x_Kw_{j_{K+1}}z'
  \end{align*}
  where each of the factors separated by the $\cdot$'s belongs
  to~$L_n[\alpha]^*$. Hence $xy^{N+1}z\in L_n[\alpha]^*$.
\end{proof}

We are now ready to complete the proof of our key result, namely that,
for $\alpha$ in normal form and $n\geq\mu[\alpha]$, the languages
$L_n[\alpha]$ are star-free.
    
\begin{proof}[\bf Proof of Theorem~\ref{t:star-free}]
  Let $i=\rk \alpha$. If $i=0$, then $L_n[\alpha]=\{\alpha\}$ is certainly a
  star-free language. We will therefore assume that $i\geq1$. Let
  $\alpha=\gamma_0(\delta_1) \gamma_1\cdots (\delta_r)\gamma_r$ be the normal form
  expression of~$\alpha$.

  We claim that each of the languages $L_n[\gamma_0]$, $L_n[\delta_j]$
  and $L_n[\delta_j\gamma_j]$ ($j=1,\ldots,r$) is star-free.
  Since, by the definition of the normal form, each $\delta_j$ is
  primitive and in circular normal form, we deduce by
  Lemma~\ref{l:pump-bases} that $L_n[\delta_j]^*$~is star-free.
  In view of Lemma~\ref{l:Ln}\ref{item:Ln-2} and since the set of star-free
  languages is closed under concatenation, it follows that each language
  $L_n[(\delta_j)\gamma_j]
  =L_n[\delta_j]^*L_n[\delta_j]^{n-1}L_n[\delta_j\gamma_j]$ is also
  star-free. Taking also into account Lemma~\ref{l:Ln}\ref{item:Ln-3},
  we conclude that the product
  $$L_n[\alpha]
  =L_n[\gamma_0]L_n[(\delta_1)\gamma_1]\cdots L_n[(\delta_r)\gamma_r]$$
  is star-free, as stated in the theorem.

  To prove the claim, we proceed by induction on $i\geq1$. The case
  $i=1$ is immediate since then all the $\gamma_j$ and $\delta_j$ are
  words in~$X^*$. Suppose that $i\geq2$ and assume inductively that
  the claim holds for $\omega$-terms of rank less than~$i$. Consider
  the $\omega$-term
  \begin{math}
    \alpha'=\gamma_0\delta_1\delta_1\gamma_1\cdots \delta_r\delta_r\gamma_r.
  \end{math}
  By condition~\ref{item:rank-i-3} of the definition of an $\omega$-term in
  normal form, $\alpha'$ is in normal form.
  By Lemma~\ref{l:mu-vs-expansions}, since $\alpha'\in E_2[\alpha]$,
  we have $\mu[\alpha]\geq\mu[\alpha']$.  Hence $n\geq\mu[\alpha']$ and we may apply the
  induction hypothesis to the $\omega$-term $\alpha'$ of rank $i-1\geq1$. Since $\alpha$
  is in normal form and the $\omega$-terms $\delta_j$ are Lyndon words of
  positive rank, the first letter of each $\delta_j$ is the opening
  parenthesis of an $\omega$-subterm of highest (and positive) rank.
  Hence, if $\alpha'=u_0(v_1) u_1\cdots (v_s) u_s$ 
  is the normal form expression of~$\alpha'$, then each factor
  $\gamma_0,\delta_j,\delta_j\gamma_j$ ($j=1,\ldots,r$) must be a product of some of the
  factors $u_0,(v_k),(v_k) u_k$ ($k=1,\ldots,s$). By the induction
  hypothesis, each of the languages $L_n[u_0]$, $L_n[v_k]$, and
  $L_n[v_ku_k]$ ($k=1,\ldots,s$) is star-free. By the above argument, it
  follows that so are the languages $L_n[u_0]$, $L_n[(v_k)]$, and
  $L_n[(v_k) u_k]$ ($k=1,\ldots,s$). Finally, by
  Lemma~\ref{l:Ln}\ref{item:Ln-2}, we deduce that each of the
  languages $L_n[\gamma_0],L_n[\delta_j],L_n[\delta_j\gamma_j]$ ($j=1,\ldots,r$) is star-free,
  thus proving the induction step. This proves the claim and completes
  the proof of Theorem~\ref{t:star-free}.
\end{proof}
We do not know whether the bound $n\geq\mu[\alpha]$ is optimal but we do know
that some bound is required, that is that $L_n[\alpha]$ may not be
star-free for $\alpha$ in normal form. An example is obtained by taking
$\alpha=((a) ab(b) a^2b^2)$, where $a$ and $b$ are letters.  Then
$L_1[\alpha]\cap[a^2b^2]^*=[a^2b^2a^2b^2]^+$ so that $L_1[\alpha]$ is not star-free
since $[a^2b^2]^*$ is star-free and $[a^2b^2a^2b^2]^+$ is not.

\section{\texorpdfstring{Factors of $\omega$-words over \pv A}{Factors of
    omega-words over A}}
\label{sec:factors-of-omega-words}

In this section we present further properties of the languages
$L_n[\alpha]$ and derive some applications. The main result of this section
is that every factor of an $\omega$-word over \pv A is also an $\omega$-word
over~\pv A.

Recall that given a pseudovariety \pv V, a finite semigroup $T\in\pv V$
\emph{satisfies} the pseudoidentity $u=v$, with $u,v\in\Om XV$, if, for
every continuous homomorphism $\varphi:\Om XV\to T$, we have $\varphi(u)=\varphi(v)$.  For
a finite semigroup $T$, let $\mathrm{ind}(T)$ be the smallest $\ell\geq1$
such that for some $k\geq1$ and every $s\in T$, we have
$s^{\ell+k}=s^\ell$. Equivalently, $\mathrm{ind}(T)$ is the minimum positive
integer $\ell$ such that $T$ satisfies the pseudoidentity
$x^{\omega+\ell}=x^\ell$. Note that $\mathrm{ind}(T)\leq|T|$. We begin by proving
that finite aperiodic semigroups do not separate an $\omega$-term from its
expansions of sufficiently large exponent.

\begin{Lemma}
  \label{l:trivial-in-finite-aperiodic}
  Let $\alpha\in \terms{X}$ be an $\omega$-term and let $T\in\pv A$.
  If $n\geq\ind T$ and $w\in L_n[\alpha]$, then $T$ satisfies the
  pseudoidentity~$\epsilon[\alpha]=w$.
\end{Lemma}

\begin{proof}
  Let $\varphi:\Om XA\to T$ be a continuous homomorphism. Since
  $n\geq\ind T$, for every $m\geq n$, the semigroup $T$ satisfies the
  identity $x^m=x^n$. Hence, for every word $w\in L_n[\alpha]$, we
  have $\varphi(w)=\varphi(u)$, where $u$ is the word which is
  obtained from $\alpha$ by replacing all occurrences of the $\omega$
  exponent by~$n$.
\end{proof}

Recall from Section~\ref{sec:pseudowords--words} that the topological
closures $\clos L$ and $\closv AL$ of a language $L$ in \Om XS and \Om
XA, respectively, are such that $\pj A(\clos L)=\closv AL$. The
following consequence of Lemma~\ref{l:trivial-in-finite-aperiodic}
will be useful.
\begin{Cor}
  \label{c:trivial-in-aperiodic}
  If $\alpha\in \terms{X}$ is an arbitrary $\omega$-term, then
  $$\pj A\Bigl(\bigcap_n\clos{L_n[\alpha]}\Bigr)
  =\{\epsilon[\alpha]\}
  =\bigcap_n\pj A\bigl(\clos{L_n[\alpha]}\bigr).$$
\end{Cor}

\begin{proof}
  Denote by $\partial$ the unique homomorphism of unary semigroups
  $\terms{X}\to\Om XS$ extending the identity mapping on~$X$ so that
  $\epsilon=\pj A\circ\partial$. First note that, since $\partial[\alpha]\in\clos{L_n[\alpha]}$ for every
  $n$, certainly $\epsilon[\alpha]\in\pj A\bigl(\bigcap_n\clos{L_n[\alpha]}\bigr)$,~so
  $$\{\epsilon[\alpha]\}
  \subseteq \pj A\Bigl(\bigcap_n\clos{L_n[\alpha]}\Bigr) \subseteq
  \bigcap_n\pj A \bigl(\clos{L_n[\alpha]}\bigr).
  $$
  Let $v\in\bigcap_n\pj A\bigl(\clos{L_n[\alpha]}\bigr)$. For a
  continuous homomorphism $\psi:\Om XA\to T$ onto a finite aperiodic
  semigroup $T$, let $\varphi=\psi\circ\pj A:\Om XS\to T$ and choose any
  $n\geq\ind T$. Then
  $$\psi(v)
  \in\varphi\bigl(\clos{L_n[\alpha]}\bigr)
  =\varphi(L_n[\alpha])
  =\{\varphi(\partial[\alpha])\}
  $$
  where the first equality follows from the continuity of~$\varphi$
  and the finiteness of~$T$, and the second equality is a consequence
  of Lemma~\ref{l:trivial-in-finite-aperiodic}. Since \Om XA is
  residually in~\pv A, it follows that $v=\epsilon[\alpha]$.
\end{proof}

We also have the following stronger result for $\omega$-terms in normal
form.

\begin{Thm}
  \label{t:pA-inverse-image-of-kterms}
  Let $w\in\omc XA$ and let $\alpha$ be the normal form representation
  of $w$. Then
  $$\pj A^{-1}(w)=\bigcap_n\clos{L_n[\alpha]}.$$
\end{Thm}

\begin{proof}
  The inclusion $\bigcap_n \clos{L_n[\alpha]}\subseteq \pj A^{-1}(w)$ follows from
  Corollary~\ref{c:trivial-in-aperiodic}. For the reverse inclusion,
  assuming that $v\in\Om XS$ is such that $\pj A(v)=w$, we have $\pj
  A(v)\in\pj A\bigl(\clos{L_n[\alpha]}\bigr)$ for all $n$.  Let $(v_n)_n$ be
  a sequence of words converging to~$v$ in~\Om XS.  Then $\lim v_n=w$
  in~\Om XA and so, since by Theorem~\ref{t:star-free} the set $\pj
  A\bigl(\clos{L_n[\alpha]}\bigr)$ is open and contains~$w$, by taking a
  suitable subsequence we may assume that $v_n\in\pj
  A\bigl(\clos{L_n[\alpha]}\bigr)\cap X^+=L_n[\alpha]$. Since $(L_n[\alpha])_n$ is a
  decreasing sequence of languages, it follows that $v\in\clos{L_n[\alpha]}$
  for all~$n$.
\end{proof}

We now prove the announced main result of this section which does not
apparently follow easily from McCammond's results.

\begin{Thm}
  \label{t:factors-of-kt}
  If $v\in\omc XA$ and $u\in\Om XA$ is a factor of~$v$, then $u\in\omc
  XA$.
\end{Thm}

\begin{proof}
  By symmetry, it suffices to prove the result when $u$ is a prefix of
  $v$, that is, when there exists $w\in\Om XA$ such that $uw=v$. Let
  $\alpha$ be the normal form representation of~$v$. We proceed by
  induction on $\rk\alpha$. We assume inductively that the result
  holds for all elements of~\omc XA with rank strictly smaller
  than~$\rk\alpha$.

  Since $L_n[\alpha]$ is star-free for $n\geq\mu[\alpha]$ by
  Theorem~\ref{t:star-free}, its closure $\closv A{L_n[\alpha]}$ is an
  open subset of~\Om XA. Hence, there exist sequences $(u_m)_m$ and
  $(w_m)_m$ converging respectively to $u$ and~$w$ such that
  $u_nw_n\in L_n[\alpha]$ for all $n\geq\mu[\alpha]$.

  As an $\omega$-term, $\alpha$ admits a unique factorization in the
  semigroup $\terms{X}$
  of the form $\alpha=x_0x_1x_2\cdots x_{2p-1}x_{2p}$, where each
  $x_{2i}$ is a finite word and each $x_{2i-1}$ is an $\omega$-term of
  the form $x_{2i-1}=(y_{2i-1})$. Note that we include here the case
  where $\alpha$ is a word, for which $p=0$. Since $\alpha$ is in
  normal form, each $y_{2i-1}$ is an $\omega$-term of rank less than
  $\rk\alpha$ (although not necessarily of $\rk\alpha-1$). In view of
  Lemma~\ref{l:Ln} and each relation $u_nw_n\in L_n[\alpha]$, there is
  a ``cutting'' index $c_n\in\{0,\ldots,2p\}$ and there are
  factorizations $u_n=u_n'u_n''$ and $w_n=w_n'w_n''$ such that
  $$u_n'\in L_n[x_0\cdots x_{c_n-1}],\ u_n''w_n'\in L_n[x_{c_n}],\
  w_n''\in L_n[x_{c_n+1}\cdots x_{2p}].$$
  Since the number of possible cutting indices depends only
  on~$\alpha$ and not on~$n$, there is a strictly increasing sequence
  of indices $(n_k)_k$ whose corresponding cutting indices are all
  equal to a certain fixed~$c$. By compactness of~\Om XA, one may
  further assume that the sequences $(u_{n_k}')_k$, $(u_{n_k}'')_k$,
  $(w_{n_k}')_k$, and $(w_{n_k}'')_k$ converge, say respectively to
  $u',u'',w',w''$. By continuity of multiplication, and since
  $(L_n[\beta])_n$ is a decreasing sequence of languages for every
  $\omega$-term $\beta$, it follows that
  \begin{align*}
    u'&\in\bigcap_n\closv A{L_n[x_0\cdots x_{c-1}]},\\
    u''w'&\in\bigcap_n\closv A{L_n[x_c]},\\
    w''&\in\bigcap_n\closv A{L_n[x_{c+1}\cdots x_{2p}]}.
  \end{align*}
  By Corollary~\ref{c:trivial-in-aperiodic}, the preceding
  intersections are reduced respectively to the $\omega$-words
  $\epsilon[x_0\cdots x_{c-1}]$, $\epsilon[x_c]$, and
  $\epsilon[x_{c+1}\cdots x_{2p}]$. Hence $u',w''\in\omc XA$ and
  $u''w'=\epsilon[x_c]$. If $c$~is even, then $u''$ is a prefix of the
  word~$x_c$ and hence $u=u'u''\in\omc XA$, as required. Hence we may
  as well assume that $\alpha$ is of the form $\alpha=(y)$.

  By Lemma~\ref{l:Ln}\ref{item:Ln-4}, we have
  $L_n[\alpha]=L_n[y^n]L_n[y]^*$. Thus, in view of the relation
  $u_nw_n\in L_n[\alpha]$, there exist factorizations $u_n=u_n'u_n''$
  and $w_n=w_n'w_n''$ such that $u_n'\in L_n[y^{r_n}]$, $u_n''w_n'\in
  L_n[y]$, and $w_n''\in L_n[y^{s_n}]$, with $r_n+s_n+1\geq n$.
  Suppose that there is a strictly increasing sequence of indices
  $(n_k)_k$ such that $r_{n_k}=r$ is constant. We may assume that the
  sequences $(u_{n_k}')_k$, $(u_{n_k}'')_k$, $(w_{n_k}')_k$, and
  $(w_{n_k}'')_k$ converge, say respectively to $u',u'',w',w''$. As
  above, it follows that $u',w''\in\omc XA$ and $u''w'=\epsilon[y]$.
  Since $\rk y<\rk\alpha$, the induction hypothesis then implies that
  $u''$ is an $\omega$-term and, therefore so is $u=u'u''$.

  Hence we may assume that $r_n\to\infty$ as $n\to\infty$. This
  implies that $y^{r_n}\to (y)$ in \omc XA. Assuming again that
  $(u_{n_k}')_k$, $(u_{n_k}'')_k$, $(w_{n_k}')_k$, and $(w_{n_k}'')_k$
  converge respectively to $u',u'',w',w''$, we conclude that
  $u'=\epsilon[(y)]\in\omc XA$ and $u''w'=\epsilon[y]$. Invoking once more the
  induction hypothesis as above, the induction step is finally
  achieved, which proves the theorem.
\end{proof}

Note that Theorem~\ref{t:star-free} only intervenes in the above
proof to show that $\closv A{L_n[\alpha]}$ is an open subset of~\Om
XA, a property which in fact is equivalent to $L_n[\alpha]$~being
star-free.

Some applications of Theorem~\ref{t:factors-of-kt} can be found
in~\cite{Almeida&Costa&Zeitoun:2009}. It plays, in particular, an
important role in establishing the main result of that paper, namely a
characterization of pseudowords over \pv A which are given by
$\omega$-terms. Other applications of Theorem~\ref{t:factors-of-kt}
and of properties of the languages $L_n[\alpha]$, such as an algorithm
to compute the closure $\closv A{L}$ of a regular language $L$ have been
published in~\cite{ACZ:Closures:14}.

\bigskip
\footnotesize
\noindent\textit{Acknowledgments.}
This work was partly supported by the \textsc{Pessoa}
French-Portuguese project Egide-Grices 11113YM \emph{Automata,
  profinite semigroups and symbolic dynamics}.
The work leading to this paper has also been carried out within the
framework of the ESF programme ``Automata: from Mathematics to
Applications (AutoMathA)'', whose support is gratefully acknowledged.
The work of the first and second authors was supported, in part, by
the European Regional Development Fund, through the programme COMPETE,
and by the Portuguese Government through FCT -- \emph{Fundação para a
  Ciência e a Tecnologia}, respectively under the projects
PEst-C/MAT/UI0144/2011 and PEst-C/MAT/UI0013/2011.
Both the first and second authors were also supported by FCT through
the project PTDC/MAT/65481/2006, which was partly funded by the
European Community Fund FEDER.
The work of the third author was partially supported by ANR 2010 BLAN
0202 01 FREC.

\bibliographystyle{amsplain} %
\bibliography{McCam-NForms}

\providecommand{\bysame}{\leavevmode\hbox to3em{\hrulefill}\thinspace}
\providecommand{\MR}{\relax\ifhmode\unskip\space\fi MR }
\providecommand{\MRhref}[2]{%
  \href{http://www.ams.org/mathscinet-getitem?mr=#1}{#2}
}
\providecommand{\href}[2]{#2}
\begin{thebibliography}{10}

\bibitem{almeida:1990}
J.~Almeida, \emph{Implicit operations on finite $\mathcal{J}$--trivial
  semigroups and a conjecture of {I. Simon}}, J. Pure Appl. Algebra \textbf{69}
  (1990), 205--218.

\bibitem{Almeida:1994a}
\bysame, \emph{Finite semigroups and universal algebra}, World Scientific,
  Singapore, 1995, English translation.

\bibitem{Almeida:1999b}
\bysame, \emph{Hyperdecidable pseudovarieties and the calculation of semidirect
  products}, Int. J. Algebra Comput. \textbf{9} (1999), 241--261.

\bibitem{Almeida99:SomeAlgorProbl}
\bysame, \emph{Some algorithmic problems for pseudovarieties}, Publ. Math.
  Debrecen \textbf{54} (1999), no.~suppl., 531--552, Automata and formal
  languages, VIII (Salg\'otarj\'an, 1996).

\bibitem{Almeida:2003c}
\bysame, \emph{Profinite semigroups and applications}, Structural Theory of
  Automata, Semigroups, and Universal Algebra (New York) (Valery~B. Kudryavtsev
  and Ivo~G. Rosenberg, eds.), NATO Science Series II: Mathematics, Physics and
  Chemistry, vol. 207, Springer, 2005, Proceedings of the NATO Advanced Study
  Institute on Structural Theory of Automata, Semigroups and Universal Algebra,
  Montréal, Québec, Canada, 7-18 July 2003.

\bibitem{Almeida&Costa&Zeitoun:2004}
J.~Almeida, J.~C. Costa, and M.~Zeitoun, \emph{Tameness of pseudovariety joins
  involving \textsf{R}}, Monatsh. Math. \textbf{146} (2005), no.~2, 89--111.

\bibitem{Almeida&Costa&Zeitoun:2005b}
\bysame, \emph{Complete reducibility of systems of equations with respect to
  {R}}, Portugal. Math. \textbf{64} (2007), 445--508.

\bibitem{ACZ:Closures:14}
J.~Almeida, J.~C. Costa, and M.~Zeitoun, \emph{Closures of regular languages
  for profinite topologies}, Semigroup Forum (2014), 1--21.

\bibitem{Almeida&Costa&Zeitoun:2009}
\bysame, \emph{Iterated periodicity over finite aperiodic semigroups}, European
  J. Combinatorics \textbf{37} (2014), 115--149.

\bibitem{Almeida&Steinberg:2000a}
J.~Almeida and B.~Steinberg, \emph{On the decidability of iterated semidirect
  products and applications to complexity}, Proc. London Math. Soc. \textbf{80}
  (2000), 50--74.

\bibitem{Almeida&Steinberg:2000b}
\bysame, \emph{Syntactic and global semigroup theory, a synthesis approach},
  Algorithmic Problems in Groups and Semigroups (J.~C. Birget, S.~W. Margolis,
  J.~Meakin, and M.~V. Sapir, eds.), Birkhäuser, 2000, pp.~1--23.

\bibitem{Almeida&Zeitoun:AutomTheorApproac:2007}
J.~Almeida and M.~Zeitoun, \emph{An automata-theoretic approach to the word
  problem for $\omega$-terms over \textsf{R}}, Theoret. Comp. Sci. \textbf{370}
  (2007), 131--169.

\bibitem{Costa:2001}
J.~C. Costa, \emph{Free profinite locally idempotent and locally commutative
  semigroups}, J. Pure Appl. Algebra \textbf{163} (2001), 19--47.

\bibitem{Costa:Canonical-forms-free-k-semigroups:2014:a}
\bysame, \emph{Canonical forms for free $\kappa$-semigroups}, Discrete Math.
  Theor. Comput. Sci. \textbf{16} (2014), no.~1, 159--178.

\bibitem{Costa&Nogueira:2009}
J.~C. Costa and C.~Nogueira, \emph{Complete reducibility of the pseudovariety
  \textsf{LSl}}, Int. J. Algebra Comput. \textbf{19} (2009), 1--36.

\bibitem{Costa&Teixeira:2005}
J.~C. Costa and M.~L. Teixeira, \emph{Tameness of the pseudovariety
  \textsf{LSl}}, Int. J. Algebra Comput. \textbf{14} (2004), 627--654.

\bibitem{Henckell:1988}
K.~Henckell, \emph{Pointlike sets: the finest aperiodic cover of a finite
  semigroup}, J. Pure Appl. Algebra \textbf{55} (1988), no.~1-2, 85--126.

\bibitem{HuschenbettKufleitner:LIPIcs:2014:4472}
M.~Huschenbett and M.~Kufleitner, \emph{{Ehrenfeucht-Fra{\"{\i}}ss{\'e} Games
  on Omega-Terms}}, STACS '14 (Ernst~W. Mayr and Natacha Portier, eds.),
  LIPIcs, vol.~25, Schloss Dagstuhl--Leibniz-Zentrum f{\"u}r Informatik, 2014,
  pp.~374--385.

\bibitem{Lothaire:1983}
M.~Lothaire, \emph{Combinatorics on words}, Addison-Wesley, Reading, Mass.,
  1983.

\bibitem{Lothaire:2001}
\bysame, \emph{Algebraic combinatorics on words}, Cambridge University Press,
  Cambridge, UK, 2002.

\bibitem{McCammond:1991}
J.~McCammond, \emph{The solution to the word problem for the relatively free
  semigroups satisfying $t^a=t^{a+b}$ with $a\geq 6$}, Int. J. Algebra Comput.
  \textbf{1} (1991), 1--32.

\bibitem{McCammond:1999a}
\bysame, \emph{Normal forms for free aperiodic semigroups}, Int. J. Algebra
  Comput. \textbf{11} (2001), 581--625.

\bibitem{McNaughton&Papert:71}
R.~McNaughton and S.~Papert, \emph{Counter-free automata}, MIT Press, 1971.

\bibitem{Moura:2011:a}
A.~Moura, \emph{The word problem for $\omega$-terms over {{\textsf{{DA}}}}},
  Theoret. Comp. Sci. \textbf{412} (2011), no.~46, 6556 -- 6569.

\bibitem{PZ:lics14}
T.~Place and M.~Zeitoun, \emph{Separating regular languages with first-order
  logic}, {CSL-LICS'14}, 2014.

\bibitem{Rhodes&Steinberg:qt}
J.~Rhodes and B.~Steinberg, \emph{The $q$-theory of finite semigroups},
  Springer Monographs in Mathematics, Springer, 2009.

\bibitem{Schutzenberger:1965}
M.~P. {S}chützenberger, \emph{On finite monoids having only trivial subgroups},
  Inform. and Control \textbf{8} (1965), 190--194.

\end{thebibliography}

\end{document}